\documentclass[journal]{IEEEtran}
\usepackage[explicit]{titlesec}
\usepackage{multirow}
\usepackage[table]{xcolor}
\usepackage{colortbl}
\usepackage{graphicx}
\usepackage{color, soul}
\usepackage{verbatim}
\usepackage{epstopdf}
\usepackage{setspace}
\usepackage{subfigure} 
\usepackage{algorithmic}
\usepackage{amsmath}
\usepackage{amsthm}
\usepackage{textcase}
\usepackage{extarrows}
\usepackage{verbatim}
\usepackage[left]{lineno}
\usepackage{cite}
\usepackage{amsmath}
\usepackage{mathrsfs}
\newtheorem{theorem}{Theorem}

 \usepackage {mathrsfs}

\usepackage{bbm}
\usepackage[colorlinks, linkcolor=blue, citecolor=blue]{hyperref}
\usepackage{amssymb}
\newtheorem{remark}{Remark}
\newtheorem{property}{Property}
\newcommand{\RNum}[1]{\uppercase\expandafter{\romannumeral #1\relax}}
\newtheorem{Lemma}{Lemma}

\newtheorem{corollary}{Corollary}
\usepackage{stfloats}
\newtheorem{problem}{Problem}
\newtheorem{proposition}{Proposition}
\usepackage{color}
\usepackage{tikz,xcolor}

\newcommand{\finished}{\hfill$\blacksquare$}
\definecolor{lime}{HTML}{A6CE39}
\titlespacing{\section}{0pt}{1.2ex plus .0ex minus .0ex}{.3ex plus .0ex}
\titlespacing{\subsection}{0pt}{1.2ex plus .0ex minus .0ex}{.3ex plus .0ex}

\makeatletter
\DeclareRobustCommand{\orcidicon}{%
	\begin{tikzpicture}
	\draw[lime, fill=lime] (0,0) 
	circle [radius=0.16] 
	node[white] {{\fontfamily{qag}\selectfont \tiny ID}};    \draw[white, fill=white] (-0.0625,0.095)
	circle [radius=0.007];    \end{tikzpicture}
	\hspace{-2mm}}
\foreach \x in {A, ..., Z}{%
	\expandafter\xdef\csname orcid\x\endcsname{\noexpand\href{https://orcid.org/\csname orcidauthor\x\endcsname}{\noexpand\orcidicon}}
}

\newcommand*\bigcdot{\mathpalette\bigcdot@{.5}}
\newcommand*\bigcdot@[2]{\mathbin{\vcenter{\hbox{\scalebox{#2}{$\m@th#1\bullet$}}}}}
\makeatother

\usepackage[linesnumbered,ruled,vlined]{algorithm2e}

\definecolor{mydarkblue}{RGB}{65, 105, 225}

\begin{document}
\title{Unified Upper Bounds on the ML decoding Error Probability of Spinal Codes over Fading Channels}
	\author{Aimin Li\orcidA{},
			\emph{Graduate Student Member, IEEE,}
	Xiaomeng Chen\orcidC{}, Shaohua Wu\orcidB{}, 
	\emph{Member, IEEE},\\
	Gary C.F. Lee\orcidI{}, 
	\emph{Member, IEEE},
	and Sumei Sun\orcidF{},
	\emph{Fellow, IEEE.} \vspace{-10mm}
\thanks{	
 An earlier version of this paper was presented in part at the 2023 IEEE International Symposium on Information Theory (IEEE ISIT 2023) \cite{chen2023tight}. This work has been supported in part by the Guangdong Basic and Applied Basic Research Foundation under Grant no. 2022B1515120002, and in part by the National Natural Science Foundation of China under Grant no. 61871147, and in part by the Major Key Project of PCL under Grant no. PCL2024A01.  (\textit{Corresponding author: Shaohua Wu.})

%
A. Li and X. Chen contribute equally to this work, and are with the Department of Electronic Engineering, Harbin Institute of Technology (Shenzhen), Shenzhen, China 518055. E-mail: liaimin@stu.hit.edu.cn; 23s052026@stu.hit.edu.cn. 

S. Wu is with the Department of Electronic Engineering, Harbin Institute of Technology (Shenzhen), Shenzhen, China 518055, and also with the Peng Cheng Laboratory, 518055, China. E-mail:  hitwush@hit.edu.cn. 

Gary C.F. Lee and Sumei Sun are with the Institute for Infocomm Research (I$^2$R), Agency for Science, Technology and Research (A*STAR), Singapore 138632. Email:  Gary\_Lee@i2r.a-star.edu.sg.
sunsm@i2r.a-star.edu.sg;
 
%
%
}
}

\maketitle
\allowdisplaybreaks
\begin{abstract}
\textcolor{black}{Performance evaluation of particular channel coding has been a significant topic in coding theory, often involving the use of bounding techniques. This paper focuses on the new family of \textit{capacity-achieving} codes, Spinal codes, to provide a comprehensive analysis framework to tightly upper bound the block error rate (BLER) of Spinal codes in the finite block length (FBL) regime. First, we resort to a variant of the \textit{Gallager random coding bound} to upper bound the BLER of Spinal codes over the fading channel. Then, this paper derives a new bound without resorting to the use of \textit{Gallager random coding bound}, achieving provable tightness over the wide range of signal-to-noise ratios (SNR). The derived BLER upper bounds in this paper are generalized, facilitating the performance evaluations of Spinal codes over different types of fast fading channels. Over the {Rayleigh}, {Nakagami-m}, and {Rician} fading channels, this paper explicitly derived the BLER upper bounds on Spinal codes as case studies. Based on the bounds, we theoretically reveal that the \textit{tail transmission pattern} (TTP) for ML-decoded Spinal codes keeps optimal in terms of reliability performance. Simulations verify the tightness of the bounds and the insights obtained.}
\end{abstract}

\begin{IEEEkeywords}
Spinal codes, block error rate (BLER), fading channels, ML decoding, upper bounds, finite block length.
\end{IEEEkeywords}
\IEEEpeerreviewmaketitle

\section{Introduction}
\subsection{Background}
First proposed in 2011 \cite{2011Spinal}, Spinal codes are a new family of \textit{capacity-achieving rateless} codes \cite{balakrishnan2012randomizing}. The \textit{capacity-achieving} and \textit{rateless} properties enable Spinal codes with superior performance in ensuring reliable and high-efficiency communications over time-varying channels. In \cite{2012Spinal}, it has demonstrated that Spinal codes outperform Raptor codes \cite{Raptor,RCONC}, Strider codes \cite{Strider} and \textit{rateless} Low-Density Parity-Check (LDPC) codes \cite{LDPC} in terms of \textit{throughput} across a wide range of channel conditions and message sizes.

Owing to the superior \textit{rateless} and \textit{capacity-achieving} properties, Spinal codes have garnered substantial attention in the realm of coding design, leading to a plethora of research endeavors including Spinal coding structure design \cite{UEPspinal,7873303,yang2016two}, high-efficiency decoding mechanisms \cite{8653979,yang2015low}, compressive Spinal codes \cite{compressive2019}, punctured Spinal codes \cite{xu2019optimized,li2020spinal}, timeliness-oriented Spinal codes \cite{meng2021analysisaoi, meng2022analysis}, and Polar-Spinal concatenation codes\cite{xu2018high,liang2020raptor,dong2017towards,cao2022polar}. These studies offer deeper insights into Spinal codes. Yet, the \textit{theoretical} analysis, especially within the Finite Block Length (FBL) regime, remains nascent, which constrains their further advancement.

\subsection{Related Works and Motivations}
In coding theory, obtaining a closed-form expression for the block error rate (BLER) of channel codes in the FBL regime is significant. Such expressions facilitate accurate performance evaluations and highlight improvements in coding design. {However, obtaining exact closed-form expressions is usually challenging, arising from the intricate, typically non-linear operations involved in the channel coding process.} As an alternative, bounds are derived for performance evaluations \cite{NBLLBC}. Today, many tight bounds have been derived, including upper bounds on Polar codes \cite{polar1,polar2}, Turbo codes \cite{681398}, Raptor codes \cite{raptorerror}, LT codes \cite{LTerror1,LTerror2}, and the upper and lower bounds on the error probability of Maximum Likelihood (ML)-decoded linear codes \cite{CIT-009}.
However, Spinal codes, a new candidate of \textit{capacity-achieving} codes, remains relatively unexplored in terms of deriving \textit{tight, explicit} bounds.

Some works have conducted theoretical analysis of Spinal codes over the AWGN channel and the binary symmetric channel (BSC). In \cite{balakrishnan2012randomizing}, Balakrishnan \emph{et.al.} conducted an asymptotic rate analysis of Spinal codes and proved that Spinal codes are \textit{capacity-achieving} over both the AWGN and the BSC channels. In \cite{UEPspinal}, Yu \textit{et.al.} carried out the FBL analysis of Spinal codes and derived the BLER upper bounds over the AWGN and the BSC channels. The core idea in \cite{UEPspinal} is an introduction of the \textit{Random Coding Union} (RCU) bound \cite[\textit{Theorem 33}]{polyanskiy2010channel} (over the BSC) and a relaxed version of the \textit{Gallager random coding bound} \cite[\textit{Theorem 5.6.2}]{gallager1968information} (over the AWGN) to upper bound Spinal codes. In \cite{li2021new} and \cite{10614151}, we further tightened the FBL upper bound over the AWGN channel by characterizing the error probability as the volume of a hypersphere divided by the volume of a hypercube, improving the tightness of the bound in the high-SNR regime. However, almost all previous works are established over the BSC or AWGN channels. The FBL analysis of Spinal codes over fading channels remains a relatively unexplored area. 

{An exception work is \cite{li2021spinal}, where the \textit{Chernoff bound} is utilized to derive an upper bound on the BLER of Spinal codes, focusing specifically on the \textit{Rayleigh} fading channel without Channel State Information (CSI).} However, due to the {probability-convergent} nature of the \textit{Chernoff bound}, the  bound in \cite{li2021spinal} holds upon a \emph{confidence probability}, {\emph{i.e.}, it is a probabilistic upper bound rather than a deterministic one}. {Consequently, the bound lacks rigorous \textit{determinism} and its {\textit{applicability}} is restricted to Rayleigh fading without CSI.} \textit{In summary, a tight, deterministic, and generalized FBL bound on the BLER of Spinal codes over fading channels remains an open challenge.}



\subsection{Main Results and Contributions}
{This paper presents two \textit{tight}, \textit{deterministic}, and \textit{unified} upper bounds on the BLER of Spinal codes over fading channels.} Building upon the earlier version in \cite{chen2023tight}, this work achieves distinctive contributions:

	 \textcolor{black}{\textbf{Theory}: (1) We derive a new bound in Theorem \ref{coretheorem1}, which is based on the variant of \textit{Gallager random coding bound}. We find that compared to the bound based on the variant of \textit{Gallager random coding bound}, our approach in Theorem \ref{coretheorem} provably achieves tighter evaluations. (2) We unified the derived bounds in \cite{chen2023tight} into a cohesive framework. Upper bounds over different fading channels are unified into a compact, generalized expression. (3) Our framework extends beyond the real-number scope of \cite{chen2023tight} to complex mapping and fading. This represents the first work that analyzes Spinal codes over complex mapping scenarios. To address this challenge, we develop new methods and tools to facilitate the analysis.}

	 \textbf{Optimization}: Building upon the theoretical analysis, we formulate a problem aimed at minimizing the BLER to optimize the transmission pattern of Spinal codes. Initially, a greedy algorithm is proposed to derive the transmission pattern. Subsequently, we find that the solutions exhibit a regularity -- invariably leading to the tail transmission pattern (TTP). Thus, we explore the optimality of the TTP and establish that the TTP is optimal for ML-decoded Spinal codes. To our knowledge, this is the first work that unveils, through theoretical proof, that transmitting tail symbols can enhance the performance of Spinal codes.

	\begin{figure}
	\centering
	\includegraphics[width=	0.88\linewidth]{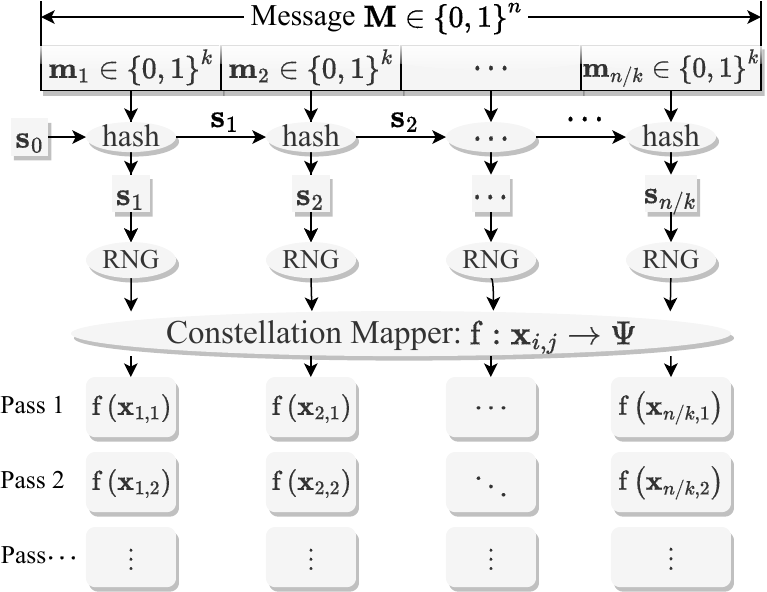}
	\caption{The encoding process of Spinal codes.}
	\label{figure1}
\end{figure}
\subsection{ Notations}
Bold symbols denote vectors or matrices. $\left\{0,1\right\}^v$ denotes a $v$-length binary sequence. $\mathfrak{R}[\cdot]$ and $\mathfrak{I}[\cdot]$ denote the real and imaginary parts of matrices, vectors, or scalars. $\exp\left\{\cdot \right\}$ represents the exponential function. $(\cdot)^\mathrm{H}$, ${(\cdot)}^*$, $\left\|\cdot\right\|_n$, $|\cdot|$, and $\Pr(\cdot)$ represents the Hermitian transpose, the complex conjugate, the $\ell^n$-norm, the modulus, and the probability. $\mathbb{P}_X(x)$ and $f_X(x)$ represent the probability mass function (PMF) and the probability density function (PDF) of the random variable $X$, respectively. $\mathbb{R}$, $\mathbb{C}$, $\mathbb{R}^L$, $\mathbb{C}^L$, and $\Psi^L$ denote the real space, complex space, $L$-dimensional real vector space, $L$-dimensional complex vector space, and $L$-fold Cartesian product of the set $\Psi$ respectively. $\mathbb{E}_X[\cdot]$ represents the expectation in terms of the random variable $X$. $\mathbb{N}$ and $\mathbb{N}^+$ denote the set of {non-negative integers} and positive integers, respectively. {$\mathcal{N}(0,\sigma^2)$ and $\mathcal{C}\mathcal{N}(0,\sigma^2)$ represent the zero-mean Gaussian distribution and the symmetric complex Gaussian distribution with variance $\sigma^2$, respectively.} $\mathbf{0}^v$ denotes the all-zero length-$v$ vector. For a positive integer $n$, $[n]$ denotes the set of integers from $1$ to $n$:  $[n]\triangleq\left\{1,2,\cdots, n\right\}$. $\Gamma(x) \triangleq \int_0^{\infty} e^{-t} t^{x-1} d t$ denotes the gamma function, $Q(x)\triangleq\frac{1}{\sqrt{2 \pi}} \int_x^{\infty} e^{-\frac{x^2}{2}}dx$ denotes the $Q$ function, and $I_0(x)$ represents the {zero-order} modified Bessel function of the first kind. 

\section{Preliminaries} \label{section II}
\subsection{Encoding Process of Spinal Codes}
This subsection introduces the encoding process of Spinal codes, as shown in Fig. \ref{figure1}. There are five key steps:
\begin{enumerate}
	\item \textit{Segmentation}: Divide an $n$-bit message $\mathbf{M}$ into $k$-bit segments $\mathbf{m}_i \in \left\{0,1 \right\}^k$, where $i\in[n/k]$.
	\item \textit{Sequentially Hashing}: The hash function\footnote{{The hash function considered in this paper possesses two critical properties: \textit{pairwise independence} and \textit{negligible collision probability} for sufficiently large $v$. The detailed properties and their proofs are provided in Appendix \ref{property}, which lays the foundation for the FBL analysis.}} $\mathcal{H}(\cdot)$ sequentially generates $v$-bit spine values $\mathbf{s}_i \in {\{0,1\}}^v$, with
	\begin{equation} \label{eqhash}
		\mathbf{s}_i = \mathcal{H}(\mathbf{s}_{i-1},\mathbf{m}_i),~i\in[n/k],~\mathbf{s}_0 = \mathbf{0}^v \text{.} 
		\footnote{The initial spine value $\mathbf{s}_0$ is known to both the encoder and the decoder and is usually set as $\mathbf{s}_0=0$ without loss of generality.}
	\end{equation}
	\item {\textit{Random Number Generator (RNG)}}: Each spine value $\mathbf{s}_i$ seeds an RNG to generate a binary pseudo-random uniform-distributed sequence $\{\mathbf{x}_{i,j}\}_{j \in \mathbb{N}^+}$. In this sequence, each $\mathbf{x}_{i,j}$ belongs to ${\{ 0,1 \}}^c$, where $c$ represents the length of $\mathbf{x}_{i,j}$. Here, $i$ is the index of spines and $j$ is the index of passes.
	\begin{equation}
		\mathrm{RNG:} ~\mathbf{s}_i \to \{\mathbf{x}_{i,j}\}_{j \in \mathbb{N}^+} ,
	\end{equation}
	\item \textit{Constellation Mapping}: The constellation mapper maps each $c$-bit symbol $\mathbf{x}_{i,j}$ to a channel input set $\Psi$:
	\begin{equation}
		\operatorname{f}: \mathbf{x}_{i,j}\rightarrow \Psi ,
	\end{equation}
	where $\operatorname{f}$ is the constellation mapping function and it converts each $c$-bit symbol $\mathbf{x}_{i,j}$ to the real space $\mathbb{R}$ or complex space $\mathbb{C}$ for transmission.
\end{enumerate}

\subsection{Channel Model} \label{section III}				
 We consider the flat fast fading channel, and thus the received symbol $y_{i,j}$ is generally expressed by
\begin{equation}
y_{i,j}=h_{i,j}\operatorname{f}(\mathbf{x}_{i,j})+n_{i,j},
\end{equation}
where $\operatorname{f}(\mathbf{x}_{i,j})\in\Psi$ is the coded symbols and $h_{i,j}$ is the corresponding fading coefficient. Under the complex-mapping constellation condition, $n_{i,j}$ follows the symmetric complex Gaussian distribution with $n_{i,j} \sim \mathcal{CN}(0,\sigma^2)$ and the distribution of $h_{i,j}$ are contingent on the type of fading channels.

\section{Bound Based on Gallager's Results} \label{sectionIV}
{Over the AWGN channel, a relaxed version of the standard \textit{Gallager random coding bound} \cite[\textit{Theorem 5.6.2}]{gallager1968information} was introduced in \cite[\textit{Theorem 4}]{UEPspinal} to upper bound the BLER of Spinal codes. However, a comprehensive FBL analysis of Spinal codes over fading channels remains an open research challenge.} 

{In this section, we extend the analytical framework developed for AWGN channels \cite[\textit{Theorem 4}]{UEPspinal} to fading channels in a two-step manner. We first derive a relaxed Gallager bound over the fading channel by leveraging an extension of \textit{Gallager bound} derived in \cite[\textit{Eq. (14)}]{karadimitrakis2017gallager}, which is detailed in subsection \ref{subA}. Then, building upon the derived relaxed Gallager Bound for fading channels, we establish an upper bound on the BLER of Spinal codes over fading channels in subsection \ref{subb}.}
\subsection{Gallager Bound over the Fading Channel}\label{subA}
In this subsection, we derive an {explicit} \textit{Gallager bound} over the general flat fast fading channel.
\begin{theorem}\label{gallagerbound1}{(Relaxed Gallager Bound for Fading Channels) For channel codes with codelength $L$, code rate $R$, and channel input set $\Psi$ transmitted over the fading channel with AWGN variance $\sigma^2$, the average BLER under ML decoding with perfect channel state information (CSI) is upper bounded by:}
\begin{equation}\label{11}
	\begin{aligned}
	&\Pr\left\{\mathcal{E}\right\}\le \\ &2^{LR}\cdot\left\{\mathbb{E}_H\left[\sum_{\beta,\alpha \in \Psi}\mathcal{Q}(\alpha)\mathcal{Q}(\beta)\exp\left\{-\frac{|H(\beta-\alpha)|^2}{4\gamma\sigma^2}\right\}\right]\right\}^L,
	\end{aligned}
\end{equation}
with $\gamma=1$ for complex fading channels and $\gamma=2$ for real fading channels\footnote{{Real fading channels typically appear in baseband systems employing real-valued modulation schemes such as Pulse Amplitude Modulation (PAM). In these scenarios, both the channel fading coefficient $h_{i,j}$ and the transmitted symbol $\operatorname{f}(\mathbf{x}_{i,j})$ are real numbers, leading to a simplified channel model $y_{i,j} = h_{i,j}\operatorname{f}(\mathbf{x}_{i,j}) + n_{i,j}$, where $n_{i,j}$ follows a real Gaussian distribution $\mathcal{N}(0,\sigma^2)$. This represents a special case of our general complex fading analysis where the imaginary components are zero.}}. Here $\mathcal{E}$ represents the event of decoding error, $H\in\mathbb{C}$ denotes a generic fading coefficient identically distributed as each $h_{i,j}$, $\mathcal{Q}(\cdot)$ denotes
the probability distribution of the channel input set $\Psi$, and $\beta, \alpha \in \Psi$ are arbitrary symbols from the channel input set $\Psi$. The double summation over $\beta$ and $\alpha$ accounts for all possible symbol pairs in the channel input alphabet $\Psi$.
\end{theorem}
\textit{Proof Sketch}. We leverage the standard \textit{Gallager bound} for a specific fading coefficient $H$ given in \cite[\textit{Eq. (7.3.20)}]{gallager1968information}, which upper bounds $\Pr\{\mathcal{E}|H\}$. Then, we express the overall BLER by averaging over all possible fading coefficients: $\Pr\{\mathcal{E}\} = \mathbb{E}_H[\Pr\{\mathcal{E}|H\}]$. For complex fading channel, the expectation $\mathbb{E}_H[\cdot]$ is solved by factoring the distribution into real and imaginary components and applying algebraic transformations. For real fading channel, the expectation $\mathbb{E}_H[\cdot]$ can be solved from the distribution of the channel fading cofficient $H$. The detalied proof is provided in Appendix \ref{bbb}.
\finished
\subsection{Upper Bound on the BLER of Spinal Codes}\label{subb}
With (\ref{11}) in hand, we can derive the upper bound on the BLER of Spinal codes over fading channels.
\begin{theorem}
	\label{coretheorem1}
	{(Upper Bound Based on Gallager's Reslut)} Consider $(n,k,c,\Psi,v\gg0)$  Spinal codes transmitted over a flat fast fading channel with AWGN variance $\sigma^2$\footnote{In the sequel, we use the shorthand $(n,k,c,\Psi,v\gg0)$ Spinal codes {to denote Spinal codes with message length $n$, segmentation parameter $k$, modulation parameter $c$, channel input set $\Psi$, and sufficiently large hash parameter $v$.}}, the average BLER under ML decoding with perfect CSI is upper bounded by
	\begin{equation}\label{14}
	P_e^{\text{Gallager}} = 1-\prod_{a\in[n/k]} \left(1-\epsilon_a^{\text{Gallager}}\right) ,
	\end{equation}
	{where $\epsilon_a^{\text{Gallager}}$ is given by:}
	\begin{equation} \label{epsilona11}
	\begin{aligned}
			\epsilon_a^{\text{Gallager}} =& 2^{k(n/k-a+1)-2cL_a}\times\\ &\left\{\mathbb{E}_H\left[\sum_{\beta,\alpha \in \Psi}\exp\left\{-\frac{|H(\beta-\alpha)|^2}{4\gamma\sigma^2}\right\}\right]\right\}^{L_a},
			\end{aligned}
	\end{equation}
	with $\gamma=1$ for complex fading channels and $\gamma=2$ for real fading channels. 
	Here $H$ characterizes the channel fading coefficient and $L_a = \sum_{i=a}^{n/k}\ell_i$, $\ell_i$ is the number of transmitted symbols generated from the spine value $\mathbf{s}_i$.
\end{theorem}
\noindent\textit{Proof.}
	The proof follows a similar structure to that presented in the proof of \cite[\textit{Theorem 4}]{UEPspinal}, with the key distinction being the substitution of \cite[\textit{Eq. (24)}]{UEPspinal} with our derived \eqref{11} in {Theorem \ref{gallagerbound1}}. This substitution directly leads to \eqref{epsilona11} and thus accomplishes the proof.
\finished

\begin{remark}
	{Our result obtained in this section is an extension of the upper bound given in \cite[\textit{Theorem 4}]{UEPspinal}. When $H\equiv1$ and $\gamma=2$, {Theorem \ref{coretheorem1}} simplifies to the bound for the AWGN channel derived in \cite[\textit{Theorem 4}]{UEPspinal}. }
\end{remark}
\section{Refined Upper Bound}
{In this section, we present a novel upper bound that provides a tighter characterization of the BLER of Spinal codes. Unlike the previous approach, this bound does not rely on Gallager's random coding technique, but leverages the inherent properties of Spinal codes to achieve greater precision. Through rigorous mathematical analysis in the proof of {Theorem \ref{compare}}, we demonstrate that the new upper bound derived in this section is provably tighter than the one established in  {Theorem \ref{coretheorem1}}.}
\subsection{Main Result}
The following theorem presents our main result on the BLER upper bound for Spinal codes. The proof is provided in subsection \ref{prooftheo4}.
\begin{theorem}
    \label{coretheorem}
(Refined Upper Bound) Consider $(n,k,c,\Psi,v\gg0)$ Spinal codes transmitted over a flat fast fading channel with AWGN variance $\sigma^2$, the average BLER given perfect CSI under ML decoding for Spinal codes can be upper bounded by
    \begin{equation}
      P_e^{\text{Spinal}} = 1-\prod_{a\in[n/k]} \left(1-\epsilon_a^{\text{Spinal}}\right) ,
    \end{equation}
 where $\epsilon_a^{\text{Spinal}}$ is
    \begin{equation}\label{epsilona}
	\epsilon_a^{\text{Spinal}} = \mathrm{min} \left\{ 1,\left(2^k-1\right)2^{n-ak} \cdot \mathscr{F} \left(L_a , \sigma,\gamma \right) \right\} ,
\end{equation}
and $\mathscr{F} \left(L_a , \sigma,\gamma \right)$ is defined as:
\begin{equation} \label{allinone}
	\begin{aligned}
		&\mathscr{F} \left(L_a , \sigma,\gamma \right)\triangleq\\&\sum_{t\in[N]} b_t\left(\sum_{\beta,\alpha \in \Psi}  2^{-2c}\mathbb{E}_{H}\left[\exp\left\{\frac{-{|H(\beta-\alpha)|}^2}{4 \gamma\sigma^2 \mathrm{sin}^2\theta_t}\right\}\right]\right)^{L_a},
	\end{aligned}
\end{equation}
{with $\gamma=1$ for complex fading channels and $\gamma=2$ for real fading channels.
Here, $H$ denotes the channel fading coefficient, $L_a = \sum_{i=a}^{n/k}\ell_i$, where $\ell_i$ is the number of transmitted symbols generated from the spine value $\mathbf{s}_i$. The sequences $\{\theta_t\}_{t=0}^N$ partitions the interval $[0,\frac{\pi}{2}]$ with $\theta_0=0$, $\theta_N=\frac{\pi}{2}$ and $\theta_0 < \theta_1 < \cdots < \theta_{N} $, the coefficient is defined as $b_t\triangleq\frac{\theta_t-\theta_{t-1}}{\pi}$, and $N$ can adjust the precision of the bound.}
\end{theorem}

{To fully understand the the bound presented in {Theorem \ref{coretheorem}}, two key questions should be addressed here:}
\begin{figure}
	\centering
	\includegraphics[width=0.8\linewidth]{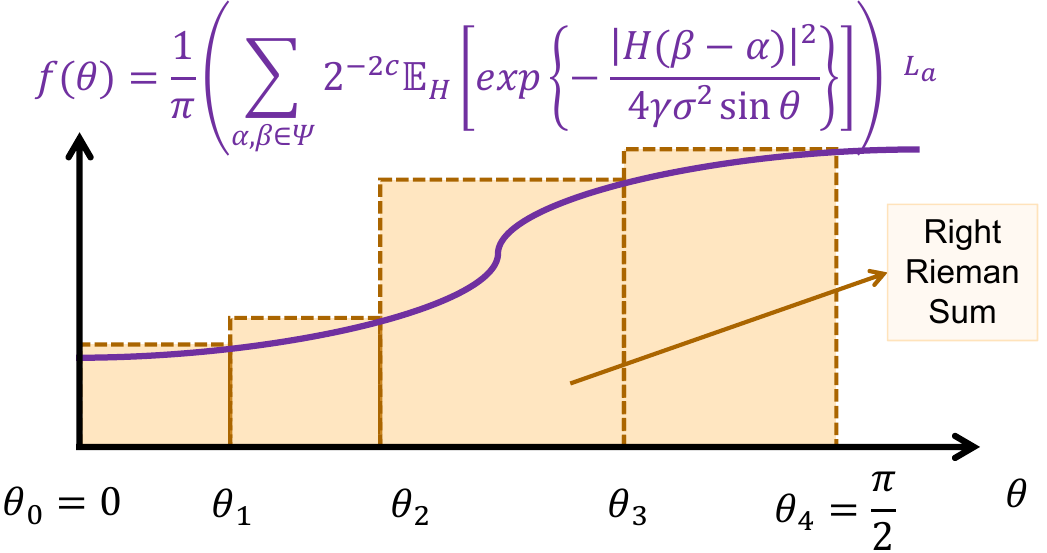}
	\caption{The rule of Right Rieman sum. }
	\label{fig:week4}
\end{figure}
\begin{itemize}
	\item {\textit{\ul{How to choose the parameters $N$ and $\{\theta_t\}_{t=0}^N$?}} The partitions $\{\theta_t\}_{t=0}^N$ implements a Riemann right sum of an integral $\int_0^{\frac{\pi}{2}}f(\theta)\mathrm{d}\theta$ that arises in the derivation. As shown in Fig. \ref{fig:week4}, the partitions $\{\theta_t\}_{t=0}^N$ can be arbitrarily chosen to obtain an upper bound of the integral. Riemann sums converge as the partition gets finer, and thus the parameter $N$ can adjust the accuracy of the bound, with the following asymptotic convergence:
	\begin{equation}
		\lim\limits_{N\to\infty}\sum_{t=1}^{N}f(\theta_t)(\theta_t-\theta_{t-1})=\int_0^{\frac{\pi}{2}}f(\theta)\mathrm{d}\theta.
	\end{equation}}
	\item \textit{\ul{How to calculate the expectation $\mathbb{E}_H[\cdot]$?}} The expectation in the bounds can be calculated in two main ways: ($i$) \textbf{Analytical Calculation}. For a specific fading channel, the paper derives closed-form expressions for the expectation in subsection \ref{casestudy} (Case Studies), which leads to explicit upper bounds. (See, e.g., Lemma \ref{theoremnaka} for Nakagami-$m$ fading channels); ($ii$) \textbf{Monte Carlo Simulation}. For more general channel models where closed-form expressions may not be available, Monte Carlo simulation can be used to estimate the expectation numerically.
\end{itemize}

The following lemma establishes a crucial property regarding the boundedness of the expectation term in our main result, ensuring that the derived upper bound is well-defined across all fading channel distributions. 
\begin{Lemma}
	The expectation term $\mathbb{E}_{H}\left[\exp\left\{\frac{-{|H(\beta-\alpha)|}^2}{4\gamma \sigma^2 \mathrm{sin}^2\theta_t}\right\}\right]$ is strictly bounded such that
	\begin{equation}
		0<\mathbb{E}_{H}\left[\exp\left\{\frac{-{|H(\beta-\alpha)|}^2}{4 \gamma\sigma^2 \mathrm{sin}^2\theta_t}\right\}\right]\le1.
	\end{equation}
\end{Lemma}
\noindent \textit{Proof.}
	We establish both the lower and upper bounds separately: For the lower bound, note that the exponential term 
	 $\exp\left\{\frac{-{|\beta-\alpha|}^2}{4\gamma \sigma^2 \mathrm{sin}^2\theta_t}\right\}$ is strictly positive for any values of $H$. Since the expectation of a strictly positive random variable is itself strictly positive, we have $\mathbb{E}_{H}\left[\exp\left\{\frac{-{|H(\beta-\alpha)|}^2}{4 \gamma\sigma^2 \mathrm{sin}^2\theta_t}\right\}\right]>0$. For the upper bound, observe that $|H(\beta-\alpha)|^2\ge0$, we have $\exp\left\{\frac{-{|\beta-\alpha|}^2}{4\gamma \sigma^2 \mathrm{sin}^2\theta_t}\right\}\le1$. A fundamental property of expectations states that if a random variable is bounded above by some constant $c_{\max}$
	 , then its expectation is also bounded above by $c_{\max}$. Applying this principle yields the upper bound $\mathbb{E}_{H}\left[\exp\left\{\frac{-{|H(\beta-\alpha)|}^2}{4 \gamma\sigma^2 \mathrm{sin}^2\theta_t}\right\}\right]\le1$.
\finished

Having established the boundedness of the expectation term, we now present a significant property that helps simplify the computation of the upper bound. The following corollary demonstrates that the phase component of the complex fading coefficient does not affect the BLER performance, allowing us to focus exclusively on the magnitude distribution when evaluating the bound.
\begin{corollary}\label{coro1}
	The bound on BLER of Spinal codes correlates solely with the magnitude distribution of the fading {coefficient}.
\end{corollary}
\noindent\textit{Proof.}
	Let us decompose the complex fading coefficient $H$ into its polar form $H=Re^{j\alpha}$. The expectation term can then be written as:
	\begin{equation}\label{50}
		\begin{aligned}
				&\mathbb{E}_{H}\left[\exp\left\{\frac{-{|H(\beta-\alpha)|}^2}{4 \gamma\sigma^2 \mathrm{sin}^2\theta_t}\right\}\right]\\&=\mathbb{E}_{R,\alpha}\left[\exp\left\{\frac{-{|Re^{j\alpha}(\beta-\alpha)|}^2}{4 \gamma\sigma^2 \mathrm{sin}^2\theta_t}\right\}\right]\\
				&=\mathbb{E}_{R}\left[\exp\left\{\frac{-{R^2|\beta-\alpha|}^2}{4 \gamma\sigma^2 \mathrm{sin}^2\theta_t}\right\}\right].
			\end{aligned}
	\end{equation}
\finished
\subsection{The Refined Upper Bound is Tighter}\label{realfadingb}
	In this subsection, we theoretically demonstrate that the refined upper bound proposed in {Theorem \ref{coretheorem1}} demonstrates greater tightness than the upper bound based on Gallager's results in {Theorem \ref{gallagerbound1}}.
	\begin{theorem}\label{compare}
		Under identical parameter setting of Spinal codes, fading channels, and SNR, {Theorem \ref{coretheorem}} yields a tighter upper bound compared to {Theorem \ref{coretheorem1}}. Specifically:
		\begin{equation}\label{12}
			P_e^{\text{Spinal}}<P_e^{\text{Gallager}}.
		\end{equation}
	\end{theorem}
	\textit{Proof Sketch.} 
	To begin with, it is easy to verify that
	\begin{equation}
		\frac{\mathrm{d}P_e^{\text{Spinal}}}{\mathrm{d}\epsilon_a^{\text{Spinal}}}\ge0,\frac{\mathrm{d}P_e^{\text{Gallager}}}{\mathrm{d}\epsilon_a^{\text{Gallager}}}\ge0, \forall a\in[\frac{n}{k}].
	\end{equation}
	 Thus, it is sufficient to prove  $\epsilon_a^{\text{Spinal}}<\epsilon_a^{\text{Gallager}},\forall a\in[\frac{n}{k}]$ to obtain \eqref{12}. Our remaining focus is to show that this inequality is true, which is accomplished by a \textit{two-step approach}:
	 
	  ($i$) We introduce an intermediate bound $\epsilon_a^{\text{Mid}}$
	  that serves as an upper bound for $\epsilon_a^{\text{Spinal}}$, with
	  \begin{equation}\label{proof:theobound1}
	  	\epsilon_a^{\text{Spinal}}\le\epsilon_a^{\text{Mid}}, \forall a\in[\frac{n}{k}].
	  \end{equation} Here, $\epsilon_a^{\text{Mid}}$ is given as:
	  \begin{equation}\label{eq:tighter}
	  	\begin{aligned}
	  		\epsilon_a^{\text{Mid}}=&\left(2^k-1\right)2^{n-ak} \cdot  \\&\frac{1}{2} \left(\sum_{\beta,\alpha \in \Psi}  2^{-2c}\mathbb{E}_{H}\left[\exp\left\{\frac{-{|H(\beta-\alpha)|}^2}{4 \sigma^2}\right\}\right]\right)^{L_a},
	  	\end{aligned}
	  \end{equation}
	  and the detailed proof of \eqref{proof:theobound1} is given in Appendix \ref{appendixd}.
	  
	  ($ii$) We next show that $\epsilon_a^{\text{Mid}}$ is upper bounded by $\epsilon_a^{\text{Gallager}}$:
	  \begin{equation}\label{proof:theobound2}
	  		\epsilon_a^{\text{Mid}}\le\epsilon_a^{\text{Gallager}}, \forall a\in[\frac{n}{k}].
	  \end{equation}
	  This is due to the fact that $(2^k-1)/2<2^k$ holds for $k>0$, which yields the following inequality:
\begin{equation}
	\begin{split}
		\epsilon_a^{\text{Mid}} &< 2^{n-ak+k} \cdot \biggl(\sum_{\beta,\alpha\in \Psi} 2^{-2c} \\
		&\quad\quad\quad \times \mathbb{E}_{H}\Bigl[\exp\Bigl\{-\frac{|H(\beta-\alpha)|^2}{4\sigma^2}\Bigr\}\Bigr]\biggr)^{L_a} = \epsilon_a^{\text{Gallager}}.
	\end{split}
\end{equation}
Combing \eqref{proof:theobound1} with \eqref{proof:theobound2} accomplishes the proof that $\epsilon_a^{\text{Spinal}}\le\epsilon_a^{\text{Mid}}<\epsilon_a^{\text{Gallager}}$, which leads to the inequality $P_e^{\text{Spinal}}<P_e^{\text{Gallager}}$.
\subsection{Proof of Theorem \ref{coretheorem}}\label{prooftheo4}
We prove this theorem by examining two cases:
\begin{itemize}
	\item Case I: Complex fading channels where $\gamma=1$.
	\item Case II: Real fading channels where $\gamma=2$.
\end{itemize}
\subsubsection{Case I: Complex Fading Channels}
	Suppose a message $ \mathbf{M}^*=\left(\mathbf{m}_1^*,\mathbf{m}_2^*,\cdots,\mathbf{m}_{n/k}^*\right)\in {\{0,1\}}^n $ is encoded to $\operatorname{f} \left( \mathbf{x}_{i,j}(\mathbf{M}^*) \right)\in\mathbb{C}$ to be transmitted over a flat fast complex fading channel. At the receiver, the ML rule given perfect CSI is	
	\begin{equation} \label{eq7}
		\begin{split}
			\widehat{\mathbf{M}} \in \underset{\mathbf{{M}} \in {\{0,1\}}^n}{\mathrm{arg\,min}}\mathscr{D}({\mathbf{M}}) ,
		\end{split}
	\end{equation}
	where $\mathscr{D}(\cdot)\triangleq\sum_{i\in[n/k]} \sum_{j\in [\ell_i]} {|y_{i,j}-h_{i,j}\operatorname{f}(\mathbf{x}_{i,j}({\cdot}))|}^2$ is the decoding cost function, and $\widehat{\mathbf{M}}=\left(\widehat{\mathbf{m}_1},\widehat{\mathbf{m}_2},\cdots,\widehat{\mathbf{m}_{n/k}}\right)\in {\{0,1\}}^n$ is the decoding result. The ML decoder aims at selecting the one with the lowest decoding cost from the candidate sequence space ${\{0,1\}}^n$. If $\widehat{\mathbf{M}}=\mathbf{M}^*$, the decoding result is correct; otherwise, a decoding error occurs.
	
The candidate sequence space ${\{0,1\}}^n$ can be partitioned into two disjoint sets: the correct decoding sequence $\mathbf{M}^*$, and the set of incorrect decoding sequences denoted as $\mathbf{M}' = (\mathbf{m}'_1,\mathbf{m}'_2,\cdots,\mathbf{m}'_{n/k}) \in \mathcal{W}$, with $\mathcal{W} \triangleq \{ (\mathbf{m}'_1,\mathbf{m}'_2,\cdots,\mathbf{m}'_{n/k}) : \exists 1 \leq i \leq n/k, \mathbf{m}'_i \neq \mathbf{m}^*_i \}$. Given $\mathbf{M}^*$ transmitted, the received signal is $y_{i,j}=h_{i,j}\operatorname{f}(\mathbf{x}_{i,j}(\mathbf{M^*}))+n_{i,j}$. The decoding cost for $\mathbf{M}^*$ is

	\begin{align} \label{eqcost1}
		\begin{split}
			\mathscr{D}(\mathbf{M}^*) &\overset{\triangle}{=} \sum_{i=1}^{n/k}\sum_{j=1}^{\ell_i}{|y_{i,j}-h_{i,j}\operatorname{f}(\mathbf{x}_{i,j}(\mathbf{M}^*))|}^2 = \sum_{i=1}^{n/k}\sum_{j=1}^{\ell_i} |n_{i,j}|^2 .
		\end{split}
	\end{align}
	While the decoding cost for an incorrect sequence $\mathbf{M}'$ is
	\begin{align} \label{eqcost2}
		\mathscr{D}(\mathbf{M}') \overset{\triangle}{=} \sum_{i=1}^{n/k}\sum_{j=1}^{\ell_i}{|y_{i,j}-h_{i,j}\operatorname{f}(\mathbf{x}_{i,j}(\mathbf{M}'))|}^2.
	\end{align}
	Let $\mathcal{E}_a$ be the event that there exists an error in the $a^{th}$ segment, \emph{i.e.}, $\widehat{\mathbf{m}_a}\ne\mathbf{m}^*_a$. Denote $\overline{\mathcal{E}}_a$ as the complement of $\mathcal{E}_a$. The BLER of Spinal codes is expressed as:
	\begin{equation} \label{eq10}
		\resizebox{1\hsize}{!}{$
			\begin{aligned}
				P_e \triangleq \Pr\left(\widehat{\mathbf{M}}\ne\mathbf{M}^*\right)&= \mathrm{Pr}\left( \bigcup_{a=1}^{n/k}\mathcal{E}_a \right) = 1 - \mathrm{Pr}\left(\bigcap_{a=1}^{n/k} \overline{\mathcal{E}}_a \right)\\
				&=1 - \prod_{a=1}^{n/k} \left[ 1 - \mathrm{Pr} \left( \mathcal{E}_a\bigg|\bigcap_{i=1}^{a-1}\overline{\mathcal{E}}_i \right)\right].
			\end{aligned}$}
	\end{equation}
	The next step is to calculate the conditional probability 
	$\mathrm{Pr} \big( \mathcal{E}_a\big|\bigcap_{i=1}^{a-1}\overline{\mathcal{E}}_i \big)$. We define $\mathcal{W}_a \triangleq \{\left(\mathbf{m}'_1, \cdots ,\mathbf{m}'_a\right)\text{:} \mathbf{m}'_1=\mathbf{m}_1^*,\cdots,\mathbf{m}'_{a-1}=\mathbf{m}_{a-1}^*,\mathbf{m}'_a \neq \mathbf{m}_a^*\} \subseteq \mathcal{W}$, capturing sequences matching the correct sequence in the first $a-1$ segments but differing in the $a$-th. The conditional probability thus reflects the chance of any sequence in $\mathcal{W}_a$ having a lower decoding cost than $\mathbf{M}^*$:
	\begin{equation} \label{eq11}
		\mathrm{Pr} \left( \mathcal{E}_a\bigg|\bigcap_{i=1}^{a-1}\overline{\mathcal{E}}_i \right) = \mathrm{Pr} \left( \exists \mathbf{M}'\in \mathcal{W}_a : \mathscr{D}(\mathbf{M}') \leq \mathscr{D}(\mathbf{M}^*) \right) \text{.}
	\end{equation}
	Applying the union bound of probability yields
	\begin{equation}\label{eq12}
		\begin{aligned}
			\mathrm{Pr} \left( \mathcal{E}_a\bigg|\bigcap_{i=1}^{a-1}\overline{\mathcal{E}}_i \right) 
			\le \min\left\{1,\sum_{\mathbf{M}' \in \mathcal{W}_a} \mathrm{Pr} \left( \mathscr{D}(\mathbf{M}') \leq \mathscr{D}(\mathbf{M}^*) \right)\right\} \text{.}
		\end{aligned}
	\end{equation}
	Our next focus of the proof is to analyze the probability $\mathrm{Pr} \left( \mathscr{D}(\mathbf{M}') \leq \mathscr{D}(\mathbf{M}^*) \right)$ in \eqref{eq12}. Through a sequence of transformations, we derive the following result:
		\begin{align}  
				&\mathrm{Pr} \left( \mathscr{D}(\mathbf{M}') \leq \mathscr{D}(\mathbf{M}^*) \right)\notag\\ &\overset{(a)}{=}\mathrm{Pr}\bigg( \sum_{i=1}^{n/k}\sum_{j=1}^{\ell_i}{|y_{i,j}-h_{i,j}f(\mathbf{x}_{i,j}(\mathbf{M'}))|}^2 \leq \sum_{i=1}^{n/k}\sum_{j=1}^{\ell_i} |n_{i,j}|^2 \bigg) \notag\\
				& \overset{(b)}{=} \mathrm{Pr}\bigg( \sum_{i=a}^{n/k}\sum_{j=1}^{\ell_i}{|y_{i,j}-h_{i,j}f(\mathbf{x}_{i,j}(\mathbf{M'}))|}^2 \leq \sum_{i=a}^{n/k}\sum_{j=1}^{\ell_i} |n_{i,j}|^2 \bigg) \notag\\
				&	\stackrel{(c)}{=}	\mathrm{Pr}\bigg(\sum_{i=a}^{n/k}\sum_{j=1}^{\ell_i}{\bigg|V_{i,j} + n_{i,j}\bigg|}^2 \leq \sum_{i=a}^{n/k}\sum_{j=1}^{\ell_i} |n_{i,j}|^2 \bigg)\notag\\
				&\stackrel{(d)}{=}\mathrm{Pr} \left( \sum_{i=a}^{n/k}\sum_{j=1}^{\ell_i} |V_{i,j}|^2 + \sum_{i=a}^{n/k}\sum_{j=1}^{\ell_i} \left(V_{i,j}{n}_{i,j}^*+{V}_{i,j}^*n_{i,j}\right) \leq 0 \right)\notag\\
				&\stackrel{(e)}{=}	\mathrm{Pr} \left( \mathfrak{R}\left[\mathbf{V}^{L_a} {\left(\mathbf{V}^{L_a} + 2{\mathbf{N}}^{L_a}\right)}^{\mathrm{H}}\right] \leq 0 \right)\notag\\
				&\stackrel{(f)}{=}  \int_{\mathbb{C}^{L_a}} \mathrm{Pr} \biggl( \mathfrak{R}\bigl[\mathbf{v}^{L_a} 
				{\left(\mathbf{v}^{L_a} + 2{\mathbf{N}}^{L_a}\right)}^{\mathrm{H}}\bigr] \leq 0 \biggr)\notag \\
				&\quad\quad\quad\quad\quad\quad\quad\quad\quad\quad\quad\quad\quad\times \mathrm{Pr} \left( \mathbf{V}^{L_a}=\mathbf{v}^{L_a} \right) 
				{~\mathrm{d}\mathbf{v}^{L_a}}\notag\\
				& \stackrel{(g)}{=}\int_{\mathbb{C}^{L_a}}Q\left( \frac{\left\|\mathbf{v}^{L_a}\right\|_2}{\sqrt{2}\sigma} \right)\cdot\mathrm{Pr} \left( \mathbf{V}^{L_a}=\mathbf{v}^{L_a} \right) 
				{~\mathrm{d}\mathbf{v}^{L_a}}\notag\\&
				\stackrel{(h)}{=}\frac{\int_{0}^{\frac{\pi}{2}}\int_{\mathbb{C}^{L_a}}  \exp\left\{\frac{-{\left\|\mathbf{v}^{L_a}\right\|_2}^2}{4 \sigma^2 \mathrm{sin}^2\theta}\right\}\cdot\mathrm{Pr} \left( \mathbf{V}^{L_a}=\mathbf{v}^{L_a} \right) ~\mathrm{d}\mathbf{v}^{L_a}  \mathrm{d}\theta}{\pi}\notag\\
				&\stackrel{(i)}{=}\frac{\int_{0}^{\frac{\pi}{2}}\prod_{i=a}^{n/k}\prod_{j=1}^{\ell_i} \int_{\mathbb{C}} \exp \left\{\frac{-|v_{i,j}|^2}{4 \sigma^2 \mathrm{sin}^2\theta}\right\} f_{V_{i,j}}(v_{i,j})~\mathrm{d}v_{i,j}\mathrm{d}\theta}{\pi}\notag\\&\stackrel{(j)}{=}\frac{\int_{0}^{\frac{\pi}{2}}{ \left( \int_{\mathbb{C}} {\exp}{\left\{-\frac{|v_{a,1}|^2}{{4}{\sigma^2}{\sin^2}\theta}\right\}} f_{V_{a,1}}(v_{a,1}) ~\mathrm{d}v_{a,1} \right) }^{L_a}\mathrm{d}\theta}{\pi}\notag\\&\stackrel{(k)}{=}\frac{\int_{0}^{\frac{\pi}{2}}\left(\mathbb{E}_{V_{a,1}}\left[{\mathrm{exp}}{\left\{-\frac{|V_{a,1}|^2}{{4}{\sigma^2}{\sin^2}\theta}\right\}}\right]\right)^{L_a}\mathrm{d}\theta}{\pi}\notag\\&\stackrel{(l)}{=}\frac{\int_{0}^{\frac{\pi}{2}}\left(\sum_{\beta,\alpha \in \Psi} \mathcal{Q}(\alpha)\mathcal{Q}(\beta){\mathbb{E}}_H\left[{\mathrm{exp}}{\left\{-\frac{|H(\beta-\alpha)|^2}{{4}{\sigma^2}{\sin^2}\theta}\right\}}\right]\right)^{L_a}\mathrm{d}\theta}{\pi}\notag\\&\stackrel{(m)}{=}\frac{\sum_{t=1}^N\int_{\theta_{t-1}}^{\theta_t} \left(\sum_{\beta,\alpha \in \Psi}2^{-2c}{\mathbb{E}}_H\left[{\mathrm{exp}}{\left\{-\frac{H(\beta-\alpha)}{{4}{\sigma^2}{\sin^2}\theta}\right\}}\right]\right)^{L_a} \mathrm{d}\theta}{\pi}\notag\\&\stackrel{(n)}{\le}\sum_{t\in[N]}b_t\left(\sum_{\beta,\alpha \in \Psi}2^{-2c}{\mathbb{E}}_H\left[{\mathrm{exp}}{\left\{-\frac{H(\beta-\alpha)}{{4}{\sigma^2}{\sin^2}\theta_t}\right\}}\right]\right)^{L_a}\label{eq4-7}\notag\\&\stackrel{(o)}{=} \mathscr{F} \left(L_a , \sigma ,1\right).
			\end{align}	
	\noindent where \eqref{eq4-7}-($a$) establishes by
	leveraging (\ref{eqcost1}) and (\ref{eqcost2}); \eqref{eq4-7}-($b$)
	 establishes since $\operatorname{f}(\mathbf{x}_{i,j}(\mathbf{M}^*))=\operatorname{f}(\mathbf{x}_{i,j}(\mathbf{M'}))$ for $1\leq i<a$. This holds because when $1\leq i<a$, $\mathbf{M}'$ and $\mathbf{M}^*$ are identical in the first $a-1$ segments, resulting in identical hash outputs (See 
	Appendix \ref{property} for the detailed proof); \eqref{eq4-7}-($c$) is obtained by applying the equality $y_{i,j}=h_{i,j}\operatorname{f}(\mathbf{x}_{i,j}(\mathbf{M^*}))+n_{i,j}$ and introducing the variable substitution
	\begin{equation}\label{eq27v2}
		V_{i,j}={h_{i,j}(\operatorname{f}(\mathbf{x}_{i,j}(\mathbf{M}^*))-\operatorname{f}(\mathbf{x}_{i,j}(\mathbf{M'})))};
	\end{equation} \eqref{eq4-7}-($d$) is derived by decomposing the square $|V_{i,j}+n_{i,j}|^2$:
\begin{equation}
	|V_{i,j}+n_{i,j}|^2=|V_{i,j}|^2+V_{i,j}{n}_{i,j}^*+{V}_{i,j}^*n_{i,j}+|n_{i,j}|^2;
\end{equation}
\eqref{eq4-7}-($e$) is established by introducing the row vector $\mathbf{V}^{L_a}\in\mathbb{C}^{L_a}$ and $\mathbf{N}^{L_a}\in\mathbb{C}^{L_a}$, which is composed of the complex random variables $\{V_{i,j}\}_{i\in[n/k],j\in[\ell_i]}$ and $\{n_{i,j}\}_{i\in[n/k],j\in[\ell_i]}$, respectively. The equality follows from the properties of inner products in complex vector spaces:
\begin{subequations}
	\begin{align}
		\mathfrak{R}\left[\mathbf{V}^{L_a}\left(\mathbf{V}^{L_a}\right)^{\mathrm{H}}\right]&=\sum_{i=1}^{n/k}\sum_{j=1}^{L_a} \mathfrak{R}\left[V_{i,j}V_{i,j}^*\right]=\sum_{i=1}^{n/k}\sum_{j=1}^{L_a} |V_{i,j}|^2,\\
			\mathfrak{R}\left[\mathbf{V}^{L_a}\left(2\mathbf{N}^{L_a}\right)^{\mathrm{H}}\right]&=\sum_{i=1}^{n/k}\sum_{j=1}^{L_a} 2\mathfrak{R}\left[V_{i,j}n_{i,j}^*\right]\notag\\&=\sum_{i=1}^{n/k}\sum_{j=1}^{L_a}\left(V_{i,j}{n}_{i,j}^*+{V}_{i,j}^*n_{i,j}\right).
	\end{align}
\end{subequations}
{\eqref{eq4-7}-($f$) is derived because of the independence between $\mathbf{V}^{L_a}$ and $\mathbf{N}^{L_a}$; \eqref{eq4-7}-($g$) is established by the following lemma, where the proof is given in Appendix \ref{app:prooflemma1}.
	\begin{Lemma} \label{Lemma3}
		Given that $n_{i,j}$ is i.i.d complex AWGN with variance $\sigma^2$, {i.e.}, $n_{i,j}\sim\mathcal{C}\mathcal{N}(0,\sigma^2)$, the following equality holds true:
		\begin{equation} \label{eq4-24}
			\mathrm{Pr} \left( \mathfrak{R}\left[\mathbf{v}^{L_a} {\left(\mathbf{v}^{L_a} + 2{\mathbf{N}}^{L_a}\right)}^{\mathrm{H}}\right] \leq 0  \right)=Q\left( \frac{\left\|\mathbf{v}^{L_a}\right\|_2}{\sqrt{2}\sigma} \right),
		\end{equation}
		where $Q(x)\triangleq\frac{1}{\sqrt{2 \pi}} \int_x^{\infty} e^{-\frac{x^2}{2}}dx$ denotes the $Q$ function.
	\end{Lemma}	
	\noindent\eqref{eq4-7}-($h$) is established by a transformation of the $Q$ function, referred to as \textit{Craig's formula} \cite{craig1991new}:
	\begin{equation} \label{eq33}
		Q(x) = \frac{1}{\pi} \int_{0}^{\frac{\pi}{2}} \mathrm{exp} \left(\frac{-x^2}{2\mathrm{sin}^2\theta}\right) \mathrm{d}\theta,
	\end{equation}
	This transformation repositions the variables of the Q function from the integral's lower limits to the integrand, thereby simplifying the analysis; \eqref{eq4-7}-($i$) is established by adopting the i.i.d of $V_{i,j}$ (see Appendix \ref{iidv} for the proof):
	\begin{equation}\label{34}
		\Pr\left(\mathbf{V}^{L_a}=\mathbf{v}^{L_a}\right)=\prod_{i=a}^{n/k}\prod_{j=1}^{\ell_i}f_{V_{i,j}}\left(v_{i,j}\right),
	\end{equation}
	and leveraging the definition of $\ell$-2 norm, which yields:
	\begin{equation}
		\exp \left\{\frac{-\left\|\mathbf{v}^{L_a}\right\|_2^2}{4 \sigma^2 \mathrm{sin}^2\theta}\right\}=\prod_{i=a}^{n/k}\prod_{j=1}^{\ell_i}\exp \left\{\frac{-v_{i,j}^2}{4 \sigma^2 \mathrm{sin}^2\theta}\right\};
	\end{equation}
	 \eqref{eq4-7}-($j$) is derived by the i.i.d nature of $V_{i,j}$; \eqref{eq4-7}-($k$) is obtained from the definition of the expectation; \eqref{eq4-7}-($l$) holds by expanding the expectation term $\mathbb{E}_{V_{a,1}}\left[{\mathrm{exp}}{\left\{-\frac{|V_{a,1}|^2}{{4}{\sigma^2}{\sin^2}\theta}\right\}}\right]$ as in \eqref{expanding} at the top of the next page,
	 \begin{figure*}
	 	\begin{equation}\label{expanding}
	 		\begin{aligned}
	 			\mathbb{E}_{V_{a,1}}\left[{\mathrm{exp}}{\left\{-\frac{|V_{a,1}|^2}{{4}{\sigma^2}{\sin^2}\theta}\right\}}\right]&\stackrel{(a)}{=}\underset{\substack{h_{a,1}, \operatorname{f}(\mathbf{x}_{a,1}(\mathbf{M'})),  \\ \operatorname{f}(\mathbf{x}_{a,1}(\mathbf{M^*}))}}{\mathbb{E}}\Big[{\mathrm{exp}}{\left\{\frac{-|h_{a,1}(\operatorname{f}(\mathbf{x}_{a,1}(\mathbf{M}^*))-\operatorname{f}(\mathbf{x}_{a,1}(\mathbf{M'})))|^2}{{4}{\sigma^2}{\sin^2}\theta}\right\}}\Big]\\&\stackrel{(b)}{=}\sum_{\beta,\alpha \in \Psi} \mathcal{Q}(\alpha)\mathcal{Q}(\beta){\mathbb{E}}_H\left[{\mathrm{exp}}{\left\{-\frac{|H(\beta-\alpha)|^2}{{4}{\sigma^2}{\sin^2}\theta}\right\}}\right]
	 		\end{aligned}
	 	\end{equation}
	 	\hrulefill
	 \end{figure*} where \eqref{expanding}-($a$) follows directly from the definition of $V_{i,j}$ as expressed \eqref{eq27v2}. Specifically, for the case where $i=a$ and $j=1$, we have \begin{equation}
	 	V_{a,1}={h_{a,1}(\operatorname{f}(\mathbf{x}_{a,1}(\mathbf{M}^*))-\operatorname{f}(\mathbf{x}_{a,1}(\mathbf{M'})))},
	 \end{equation}
	 and \eqref{eq4-7}-($m$) leverages the fundamental property that the fading coefficient $h_{a,1}$, the encoded symbols $f(\mathbf{x}_{a,1}(\mathbf{M}'))$, and $f(\mathbf{x}_{a,1}(\mathbf{M}^*))$ are mutually independent (see Appendix \ref{iidv} for the proof) and the uniform distribution of codeword $\mathcal{Q}(x)=2^{-c},\forall x\in\Psi$; 
	 \eqref{eq4-7}-($m$) employs a standard numerical integration technique by partitioning the integration interval $[0,\frac{\pi}{2}]$ into a collection of $N$ sub-intervals $([\theta_{t-1},\theta_t]$ for $t\in[N]$, where $\theta_0=0$ and $\theta_N=\frac{\pi}{2}$; 
	 The upper bound \eqref{eq4-7}-($n$) applies the \textit{rule of Rieman right sum}. This upper bounding technique is  effective when the integrand is monotonically increasing with respect to the integration variable. The following lemma establishes this monotonicity property.
	\begin{Lemma}\label{lemma8}
		$\left(\sum_{\beta,\alpha \in \Psi}2^{-2c}{\mathbb{E}}_H\left[{\mathrm{exp}}{\left\{-\frac{H(\beta-\alpha)}{{4}{\sigma^2}{\sin^2}\theta}\right\}}\right]\right)^{L_a}$ is monotonically increasing with $\theta$.
	\end{Lemma}
	\noindent This monotonicity property ensures that evaluating the integrand at the right endpoint of each sub-interval yields a value greater than or equal to its value at any other point within that sub-interval;
	\noindent Finally, \eqref{expanding}-($o$) completes the derivation by recognizing that the resulting expression corresponds precisely to $\mathscr{F}(L_a,\sigma,1)$ as defined in \eqref{allinone}. With \eqref{eq4-7} in hand, we can then substitute it into \eqref{eq12} to establish the following bound:
	\begin{equation} \label{newww}
		\begin{split}
			\sum_{\mathbf{M}' \in \mathcal{W}_a} \mathrm{Pr} \left( \mathscr{D}\left(\mathbf{M}'\right) \le  \mathscr{D}\left(\mathbf{M}\right) \right) &\le \sum_{\mathbf{M}' \in \mathcal{W}_a} \mathscr{F}(L_a,\sigma,1)\\&={\left| \mathcal{W}_a \right|}\cdot\mathscr{F}(L_a,\sigma,1),
		\end{split}
	\end{equation}
	where $\left| \mathcal{W}_a \right| = (2^k-1)2^{n-ak}$. Substituting (\ref{newww}) into (\ref{eq12}) accomplishes the proof.}\finished
	{\subsubsection{Case II: Real Fading Channels}
	Over the real fading channel with AWGN variance $\sigma^2$, the bound in 
	\eqref{eq4-7} turns to
	\begin{align}  
		&\mathrm{Pr} \left( \mathscr{D}(\mathbf{M}') \leq \mathscr{D}(\mathbf{M}^*) \right)\notag\\ 
		&\stackrel{(a)}{=}  \int_{\mathbb{R}^{L_a}} \mathrm{Pr} \biggl( \mathbf{v}^{L_a} 
		{\left(\mathbf{v}^{L_a} + 2{\mathbf{N}}^{L_a}\right)}^{\mathrm{H}} \leq 0 \biggr)\notag \\
		&\quad\quad\quad\quad\quad\quad\quad\quad\quad\quad\quad\quad\quad\times \mathrm{Pr} \left( \mathbf{V}^{L_a}=\mathbf{v}^{L_a} \right) 
		{~\mathrm{d}\mathbf{v}^{L_a}}\notag\\
		& \stackrel{(b)}{=}\int_{\mathbb{R}^{L_a}}Q\left( \frac{\left\|\mathbf{v}^{L_a}\right\|_2}{{2}\sigma} \right)\cdot\mathrm{Pr} \left( \mathbf{V}^{L_a}=\mathbf{v}^{L_a} \right) 
		{~\mathrm{d}\mathbf{v}^{L_a}}\notag\\&\stackrel{(c)}{\le}\sum_{t\in[N]}b_t\left(\sum_{\beta,\alpha \in \Psi}2^{-2c}{\mathbb{E}}_H\left[{\mathrm{exp}}{\left\{-\frac{H(\beta-\alpha)}{{8}{\sigma^2}{\sin^2}\theta_t}\right\}}\right]\right)^{L_a}\label{eq4-9}\notag\\&\stackrel{(d)}{=} \mathscr{F} \left(L_a , \sigma,2 \right),
	\end{align}	
	where \eqref{eq4-9}-($a$) follows a similar process of \eqref{eq4-7}-($a$-$f$); \eqref{eq4-9}-($b$) is established by the following lemma:
	\begin{Lemma}
		(Restatement of \cite[\textit{Lemma 1}]{chen2023tight}) Given that $n_{i,j}$ is i.i.d AWGN with variance $\sigma^2$, {i.e.}, $n_{i,j}\sim\mathcal{N}(0,\sigma^2)$, the following equality holds true:
		\begin{equation} 
			\mathrm{Pr} \left( \mathbf{v}^{L_a} {\left(\mathbf{v}^{L_a} + 2{\mathbf{N}}^{L_a}\right)}^{\mathrm{H}} \leq 0  \right)=Q\left( \frac{\left\|\mathbf{v}^{L_a}\right\|_2}{{2}\sigma} \right).
		\end{equation}
	\end{Lemma}
	\noindent\eqref{eq4-9}-($c$) follows a similar process of \eqref{eq4-7}-($h$-$n$), and \eqref{eq4-9}-($d$) follows directly from the definition of $\mathscr{F}(L_a,\sigma,\gamma)$ given in \eqref{allinone}.
	Substituting \eqref{eq4-9} into (\ref{eq12}) accomplishes the proof.} \finished 
\section{Case Studies and Explicit Bounds}\label{casestudy}
A key advantage of our analytical framework is the ability to derive explicit, computationally efficient expressions for the expectation term  $\mathbb{E}_{R}\left[\exp\left\{\frac{-{|R(\beta-\alpha)|}^2}{4\gamma \sigma^2 \mathrm{sin}^2\theta_t}\right\}\right]$ to determine the BLER upper bound on Spinal codes. We establish a unified approach that addresses multiple fading channel models within a coherent mathematical framework.

For three widely-used fading channels, Nakagami-m, Rician, and Rayleigh (a special case of Nakagami-m where $m=1$ and Rician where $K=0$) fading channels, we developed closed-form expressions that eliminate the need for numerical integration or Monte Carlo methods when evaluating the BLER bound. In the following subsections, we present these explicit expressions for each fading channel.
{\subsection{Explicit Bounds for Nakagami-m Fading Channels}
The following lemma provides a closed-form solution for the expectation term $\mathbb{E}_{R}\left[\exp\left\{\frac{-{R^2|\beta-\alpha|}^2}{4\gamma \sigma^2 \mathrm{sin}^2\theta_t}\right\}\right]$ under Nakagami-m fading, where the proof is detailed in Appendix \ref{casestudy1}.
\begin{Lemma} \label{theoremnaka}
	Under the Nakagami-m fading channel with mean square $\Omega$, AWGN variance $\sigma^2$, and Nakagami parameter $m$, the expectation $\mathbb{E}_{R}\left[\exp\left\{\frac{-{R^2|\beta-\alpha|}^2}{4 \gamma\sigma^2 \mathrm{sin}^2\theta_t}\right\}\right]$ is closed-form:
	\begin{equation}\label{naka}
	G_{\text{Naka}}(m)={\left(\frac{4\gamma m{\sigma}^2{\sin}^2{\theta_t}}{{\Omega}|\beta-\alpha|^2+4\gamma m{\sigma}^2{\sin}^2{\theta_t}}\right)}^m.
	\end{equation}
\end{Lemma}
With {Lemma \ref{theoremnaka}} in hand, we can apply {Theorem \ref{coretheorem1}} and {Corollary \ref{coro1}} to upper bound the BLER of Spinal codes. The following corollary establishes key properties of $G_{\text{Naka}}(m)$, which reveals important connections between different fading models and provides valuable analytical insights.}
\begin{corollary}\label{coro2}
	{The following assertions hold true:}
	
	{($i$) $G_{\text{Naka}}(m)$ is monotonically decreasing with $m$;}
	
	{($ii$) For Rayleigh fading channels, we can calculate the expectation term by $G_{\text{Naka}}(1)$ to determine the upper bound;}
	
	{($iii$) For AWGN channels, we can calculate the limit $\lim\limits_{m\rightarrow\infty}G_{\text{Naka}}(m)$ to determine the upper bound, where}
	\begin{equation}\label{gausian}
			\begin{aligned}
				\lim\limits_{m\rightarrow\infty}G_{\text{Naka}}(m)=\exp\left\{-\frac{\Omega|\beta-\alpha|^2}{4\gamma\sigma^2\sin^2\theta_t}\right\}.
			\end{aligned}
	\end{equation} 
\end{corollary}
\noindent {\textit{Proof.} See Appendix \ref{proof:coro1}.} \finished
\subsection{{Explicit Bounds for Rician Fading Channels}}
{The following lemma provides a closed-form solution for the expectation term $\mathbb{E}_{R}\left[\exp\left\{\frac{-{R^2|\beta-\alpha|}^2}{4\gamma \sigma^2 \mathrm{sin}^2\theta_t}\right\}\right]$ under Rician fading, where the proof is detailed in
Appendix \ref{casestudy2}.
\begin{Lemma}
\label{theoremricecom}
Under the Rician fading channel with mean square $\Omega$, AWGN variance $\sigma^2$, and Rician factor $K$, the expectation $\mathbb{E}_{R}\left[\exp\left\{\frac{-{R^2|\beta-\alpha|}^2}{4 \gamma\sigma^2 \mathrm{sin}^2\theta_t}\right\}\right]$ is closed-form:
\begin{equation}\label{rician0}
\begin{aligned}
&G_{\text{Rician}}(K)=\frac{4\gamma(K+1){\sigma}^2{\sin}^2{\theta}_t}{{\Omega}|\beta-\alpha|^2+4\gamma(K+1){\sigma}^2{\sin}^2{\theta}_t}\times\\\qquad  &\quad\quad\quad{\mathrm{exp} \left\{ \frac{-K\Omega|\beta-\alpha|^2}{{\Omega}|\beta-\alpha|^2+4\gamma(K+1){\sigma}^2{\sin}^2{\theta}_t}\right\}}.
\end{aligned}
\end{equation}
\end{Lemma}}
{Similarly,  {Lemma \ref{theoremricecom}} enables us to apply {Theorem \ref{coretheorem1}} and {Corollary \ref{coro1}} to upper bound the BLER of Spinal codes. The following corollary reveals some key properties of $G_{\text{Rician}}(K)$.}
\begin{corollary}\label{coro3}
	{The following assertions hold true:}
	
	{($i$) $G_{\text{Rician}}(K)$ is monotonically decreasing with $K$;}
	
	{($ii$) For Rayleigh fading channels, we can calculate the expectation term by $G_{\text{Rician}}(0)$ to determine the upper bound;}
	
	{($iii$) For AWGN channels, we can calculate the limit $\lim\limits_{K\rightarrow\infty}G_{\text{Rician}}(K)$ to determined the upper bound, where
	\begin{equation}\label{awgn}
	\lim\limits_{K\rightarrow\infty}G_{\text{Rician}}(K)=\exp\left\{-\frac{\Omega|\beta-\alpha|^2}{4\gamma\sigma^2\sin^2\theta_t}\right\}.
\end{equation}} 
\end{corollary}
\noindent {\textit{Proof.} See Appendix \ref{proof:coro2}.} \finished

\section{Optimal Transmission Scheme under ML Decoding} \label{sectionnew}
This section aims at optimizing the transmission scheme of Spinal codes. We formulate a BLER minimization problem constrained by a fixed coding rate and find that for ML-decoded Spinal codes, transmitting tail symbols consistently leads to the optimal solution. This provides theoretical support for the heuristic of transmitting tail symbols, as proposed in previous literature.  
\subsection{Problem Formulation and Solution}
Leveraging the upper bound on the BLER of Spinal codes given in {Theorem \ref{coretheorem}}, denoted by $P_e^{\text{Spinal}}$, we could establish an optimization problem to optimize Spinal codes' transmission pattern. The optimization problem is explicitly given as:
\begin{problem}\label{p1}
		$\begin{aligned} \label{Optim}
			\min_{\mathbf{L}} \,& P_e^{\text{Spinal}}~~ \mathrm{s.t.} \sum_{i=1}^{n/k} \ell_{i} = N; \ell_i \in \mathbb{N}^{+}, i \in[n/k].
		\end{aligned}$ 
\end{problem}
{Problem \ref{p1} aims to find the optimal allocation of transmitted symbols among different spine values $\mathbf{L} = [\ell_1,\ell_2,\ldots,\ell_{n/k}]$ to minimize the BLER upper bound $P_e^{\text{Spinal}}$ while maintaining a fixed overall code rate. Although this integer planning problem is challenging to solve directly, our interest is primarily in exploring potential patterns in optimal solutions rather than developing efficient algorithms. As part of our preliminary investigation, we implemented a simple greedy approach (Algorithm \ref{Algorithm 1}) to observe solution characteristics. The solution's dynamics, with parameters $r=3$ and $N=19$, are demonstrated in Fig. \ref{fig2}, which revealed intriguing patterns that motivated our subsequent theoretical analysis.}
\begin{figure}
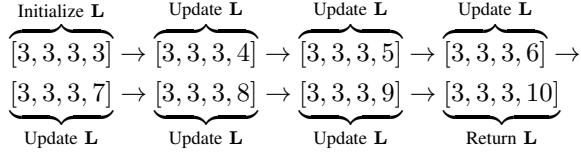

	\begin{equation}
	\begin{aligned}
	\overbrace{[3,3,3,3]}^{\text{Initialize } \mathbf{L}}&\rightarrow\overbrace{[3,3,3,4]}^{\text{Update }\mathbf{L}}\rightarrow\overbrace{[3,3,3,5]}^{\text{Update }\mathbf{L}}\rightarrow\overbrace{[3,3,3,6]}^{\text{Update }\mathbf{L}}\rightarrow\\\underbrace{[3,3,3,7]}_{\text{Update }\mathbf{L}}&\rightarrow\underbrace{[3,3,3,8]}_{\text{Update }\mathbf{L}}\rightarrow\underbrace{[3,3,3,9]}_{\text{Update }\mathbf{L}}\rightarrow\underbrace{[3,3,3,10]}_{\text{Return }\mathbf{L}}
	\end{aligned}\nonumber
	\end{equation}
	\caption{A dynamic solution process of Algorithm \ref{Algorithm 1}. Parameter is set as $n=8$, $k=2$, $r=3$, and $N=19$.} \label{fig2}
\end{figure}

\begin{algorithm}
	\caption{The greedy baseline algorithm for solving {Problem \ref{p1}}}
	\label{Algorithm 1}
	\LinesNumbered
	\KwIn{Initialize number of transmitted passes $p_{\textbf{ini}}$; Preset the target number of pass $N$ (make sure $p_{\textbf{ini}}n/k \leq N$);}
	\KwOut{The number of symbols generated from each spine value $\mathbf{L} = [\ell_1,\ell_2,\cdots,\ell_{n/k}]$;}
	Initialization: $\mathbf{L} = [p_{\textbf{ini}},p_{\textbf{ini}},\cdots,p_{\textbf{ini}}]$, $N \geq p_{\textbf{ini}}n/k$ \;
    Calculate $P_e^{\mathrm{U}}$ by applying {{Theorem \ref{coretheorem}}}\;
    \While {$\sum_{i=1}^{n/k} \ell_{i} < N$}
    {
    	\For{$i \gets 1$ to $n/k$}
    	{
    		Update the decision variable: $\ell_i \gets \ell_i+1$\;
    		Calculate BLER bound $P_{e,i}^{\mathrm{U}}$ \;
    		Restore the decision variable: $\ell_i \gets \ell_i-1$\;
    	}	
    	Search $d=\arg\min_{i}P_{e,i}^{\mathrm{U}}$\;		
    	Update $\ell_d \gets \ell_d+1$, $P_e^{\mathrm{U}} \gets P_{e,d}^{\mathrm{U}}$;
    }
    \KwRet{$\mathbf{L}$}
\end{algorithm} 

A distinct trend emerges in Fig. \ref{fig2}: the tail transmitting pattern (TTP), in which iterations consistently transmit tail symbols. {This observation led to our central theoretical discovery: the optimality of TTP for ML-decoded Spinal codes, which we rigorously prove in the next section.}

\subsection{Optimality of the TTP Scheme}\label{proof}
It's important to note that while the TTP scheme's efficacy for Spinal codes was empirically identified in \cite{2012Spinal}, a theoretical basis explaining its effectiveness remained unexplored. This work fills that gap by theoretically substantiating the TTP scheme's optimality.

\begin{theorem} \label{optiaml scheme}
The TTP scheme is optimal for Problem \ref{p1}.
\end{theorem}
\begin{proof}
We first examine the relationship between the BLER and code length given a fixed code rate. 
\begin{Lemma} \label{Lemma_optm}
	$\epsilon_a$ is non-increasing with $L_a$ for $\forall 1 \leq a \leq n/k$.
\end{Lemma}
\begin{proof}
	Note that $\epsilon_a$ is non-increasing with $\mathscr{F}(L_a,\sigma,\gamma)$, the monotony of $\epsilon_a$ w.r.t $L_a$ is equivalent to the monotony of $\mathscr{F}(L_a,\sigma,\gamma)$ w.r.t $L_a$. To discuss the monotony of $\mathscr{F}(L_a,\sigma,\gamma)$ concerning $L_a$, it's essential to ascertain whether the base of the exponential function inherent in (\ref{allinone}) exceeds $1$. This base is defined as  $ \sum_{\beta,\alpha \in \Psi} 2^{-2c} \mathbb{E}_{R}\left[\exp\left\{\frac{-{R^2|\beta-\alpha|}^2}{4\gamma \sigma^2 \mathrm{sin}^2\theta_t}\right\}\right] $. Given that  $\exp\left\{\frac{-{R^2|\beta-\alpha|}^2}{4\gamma \sigma^2 \mathrm{sin}^2\theta_t}\right\}\le1$, its expected value also holds that $\mathbb{E}_{R}\left[\exp\left\{\frac{-{R^2|\beta-\alpha|}^2}{4\gamma \sigma^2 \mathrm{sin}^2\theta_t}\right\}\right]\le 1$. Consequently, we establish the following inequality:
	\begin{equation}
	\resizebox{1\hsize}{!}{$
	\begin{aligned}
	\sum_{\beta,\alpha \in \Psi} 2^{-2c} \mathbb{E}_{R}\left[\exp\left\{\frac{-{R^2|\beta-\alpha|}^2}{4\gamma \sigma^2 \mathrm{sin}^2\theta_t}\right\}\right]\le\sum_{\beta \in \Psi} \sum_{\alpha  \in \Psi} 2^{-2c} = 1,
	\end{aligned}$}
	\end{equation}
	which indicates that $\mathscr{F}(L_a,\sigma,\gamma)$ is decreasing with $L_a$. Then, we obtain that $\epsilon_a$ is non-increasing with $L_a$.
\end{proof}
With {Lemma \ref{Lemma_optm}} in hand, we now start to prove the \textit{optimality} of the TTP. Consider two types of transmission patterns, one is the TTP pattern, denoted by a vector
\begin{equation}\label{66}
\mathbf{L}^*=(\ell_1,\ell_2,\cdots,\ell_{n/k}+M),
\end{equation}
where $(\ell_1,\ell_2,\cdots,\ell_{n/k})$ is the initialization transmission pattern. The other could be arbitrary patterns, denoted by \begin{equation}\label{52}\mathbf{L} = (\ell_1+\delta_1,\ell_2+\delta_2,\cdots,\ell_{n/k}+\delta_{n/k}), \delta_i\in\mathbb{N}.\end{equation}
To ensure that both the aforementioned patterns align with the same code rate, they must satisfy that $\left\|\mathbf{L}^*\right\|_1=\left\|\mathbf{L}\right\|_1$, which is equivalent to the condition $M = \sum_{j=1}^{n/k}\delta_i$. Subsequently, we compare the upper bounds on the BLER of Spinal codes with respect to $\mathbf{L}^*$ and $\mathbf{L}$, respectively. Denote $P_e(\mathbf{L})$ and $P_e(\mathbf{L}^*)$ as the upper bounds on the BLER of ML-decoded Spinal codes w.r.t $\mathbf{L}$ and $\mathbf{L}^*$. The optimality of the transmission pattern \(\mathbf{L}^*\) is equivalent to showing that its BLER, \(P_e(\mathbf{L}^*)\), is the lowest when compared to the BLER, \(P_e(\mathbf{L})\), of any other arbitrary patterns under the same code rate. This is mathematically described by the following {Proposition} \ref{Lemma6}. If we could prove that this proposition holds true, then we accomplish the proof of {Theorem \ref{optiaml scheme}}.
\begin{proposition}\label{Lemma6}
	For $\forall M,\delta_1,\cdots,\delta_{n/k} \in \mathbb{N}$ such that $M = \sum_{j=1}^{n/k} \delta_{j}$, $P_e({\mathbf{L}}^*) \leq P_e(\mathbf{L})$.
\end{proposition}
\begin{proof}
Performing (\ref{eq10}) yields 
\begin{equation}
\begin{aligned}
P_e(\mathbf{L}) &= 1 - \prod_{a=1}^{n/k}(1-\epsilon_a(\mathbf{L})),\\P_e({\mathbf{L}^*}) &= 1 - \prod_{a=1}^{n/k}(1-\epsilon_a({\mathbf{L}^*})).
\end{aligned}
\end{equation}
Thus, it is natural to find that $\frac{\delta P_e}{\delta \epsilon_i}=\prod_{a\in[n/k]/ i}(1-\epsilon_a)\ge0$ holds for $\forall i\in[n/k]$. This indicates that $P_e$ is increasing with $\epsilon_i$ for $\forall i\in[n/k]$. Thus, we could initially explore the stronger inequalities such that $\epsilon_a({\mathbf{L}^*}) \leq \epsilon_a(\mathbf{L})$ for $1 \leq a \leq n/k$. If these strong inequalities $\epsilon_a({\mathbf{L}^*}) \leq \epsilon_a(\mathbf{L})$ for $1 \leq a \leq n/k$ hold true, then establishing the weaker inequality $P_e({\mathbf{L}}^*) \leq P_e(\mathbf{L})$ becomes straightforward.

Fortunately, we can indeed prove that the strong inequalities $\epsilon_a(\mathbf{L}^*) \leq \epsilon_a(\mathbf{L})$ for $1 \leq a \leq n/k$ hold true. This is achieved by leveraging the monotonically decreasing nature of $\epsilon_a$ w.r.t $L_a$ in {Lemma \ref{Lemma_optm}}. If we could prove $\mathbf{L}^*_a \geq \mathbf{L}_a$, where $\mathbf{L}^*_a$ and $\mathbf{L}_a$ are defined as the cumulative transmitted symbols after the $a$-th segment of Spinal codes, $\epsilon_a(\mathbf{L}^*) \leq \epsilon_a(\mathbf{L})$ naturally holds. 

According to the definitions of $\mathbf{L}^*_a$ and $\mathbf{L}_a$, together with (\ref{66}) and (\ref{52}), we have that
\begin{equation}
\mathbf{L}_a = \sum_{i=a}^{n/k}\ell_i + \sum_{i=a}^{n/k} \delta_i,\quad
\mathbf{L}^*_a = M + \sum_{i=a}^{n/k}\ell_i.\label{86}
\end{equation} Subtract $\mathbf{L}_a$ from $\mathbf{L}^*_a$, we have
\begin{equation}
\mathbf{L}^*_a - \mathbf{L}_a = M- \sum_{i=a}^{n/k} \delta_i.
\end{equation}
Note that $M=\sum_{i\in[n/k]}\delta_i$, we have $\mathbf{L}^*_a - \mathbf{L}_a=\sum_{i=1}^{a-1}\delta_i\ge0$.
With $\mathbf{L}^*_a \ge {\mathbf{L}}_a $, we could indicate by {Lemma \ref{Lemma_optm}} that $\epsilon_a(\mathbf{L}^*) \leq \epsilon_a(\mathbf{L})$, and thus accomplish the proof such that $P_e(\mathbf{L}^*) \leq P_e(\mathbf{L})$.
\end{proof}
Now that {Proposition \ref{Lemma6}} holds true, we then accomplish the proof of {Theorem \ref{optiaml scheme}}.
\end{proof}
\begin{figure}[t]
	\centering
	\includegraphics[width=	1\linewidth]{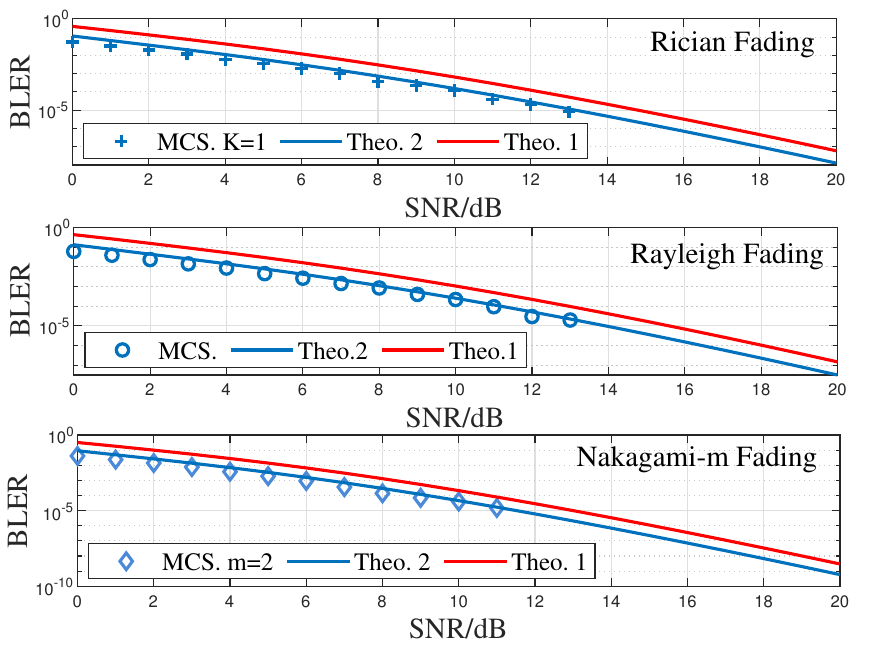}
	\caption{Upper Bounds vs. Monte Carlo Simulations (MCS).  BLER of Spinal codes with $n = 8, v = 32, pass = 6, c = 8$ and $k = 2$ over complex fading channels with $\Omega = 1$.}
	\label{figure3}
\end{figure}
\section{Simulations}\label{sectionV}
In this section, we conduct simulations to verify the obtained bounds and to validate the insights.
\subsection{Explicit Upper Bounds vs. Monte Carlo Simulations}
\begin{table*}[t]
	\centering
	\caption{{Relative Tightness Comparison Between Theoretical Bounds and Monte Carlo Simulations $\eta$}}
	\label{tab:tightness}
	\begin{tabular}{|c|c|c|c|c|c|c|}
		\hline
		\multirow{2}{*}{SNR (dB)} & \multicolumn{2}{c|}{Rayleigh Fading} & \multicolumn{2}{c|}{Rician Fading ($K=1$)} & \multicolumn{2}{c|}{Nakagami-$m$ ($m=2$)} \\
		\cline{2-7}
		& Theo. 1 (\textbf{Benchmark}) & Theo. 2 & Theo. 1 & Theo. 2 & Theo. 1 & Theo. 2 \\
		\hline
		3 & 5.44 & \cellcolor{gray!20}{\textbf{0.78} $\color{blue}\downarrow$} & 5.36 & \cellcolor{gray!20}{\textbf{0.72} $\color{blue}\downarrow$} & 6.09 & \cellcolor{gray!20}{\textbf{0.82}  $\color{blue}\downarrow$} \\
		\hline
		6 & 5.12 & \cellcolor{gray!20}{\textbf{0.59} $\color{blue}\downarrow$} & 5.23 & \cellcolor{gray!20}{\textbf{0.58} $\color{blue}\downarrow$} & 6.23 & \cellcolor{gray!20}{\textbf{0.71} $\color{blue}\downarrow$} \\
		\hline
		9 & 4.53 & \cellcolor{gray!20}{\textbf{0.36} $\color{blue}\downarrow$} & 5.18 & \cellcolor{gray!20}{\textbf{0.48} $\color{blue}\downarrow$} & 6.65 & \cellcolor{gray!20}{\textbf{0.69} $\color{blue}\downarrow$} \\
		\hline
	\end{tabular}
\end{table*}
Given the exponential complexity of ML-decoding, we opt for a relatively minimal value of $n=8$ for the message size. We set the number of passes as $L = 6$ to facilitate a manageable ML-decoding Monte Carlo simulation setup. The parameter $v$ is designated as $v = 32$, as substantiated by {Property \ref{property2}} in Appendix \ref{property}, elucidating that a hash collision is anticipated to occur once per $2^{32}$ hash function invocations on average. Additionally, we set $N=20$ and {$\theta_t=\frac{t\pi}{2N}, t\in[N]$} in {Lemma \ref{theoremnaka}} to {Lemma \ref{theoremricecom}} to ensure the precision of upper bounds approximations. The sample size for Monte Carlo simulations is set as $10^6$ to calculate the average BLER.  All channel mean square values are normalized by setting $\Omega = 1$. For complex Nakagami-m fading channels, the fading parameter is fixed at $m=2$. For complex Rician fading channels, the Rician factor is set to $K=1$.

\textcolor{black}{Fig. \ref{figure3} demonstrates that our derived bounds remain tight across a range of fading channels and SNR conditions. This illustrates our unified approach's robustness in providing tight BLER upper bounds for various fading scenarios. Notably, the approach in {Theorem \ref{coretheorem}} is tighter than that {Theorem \ref{coretheorem1}}, which is based on the \textit{Gallager bound}, consistent with our theoretical insight shown in {Theorem \ref{compare}}.}

{To quantitatively evaluate the tightness of our bounds, we introduce the \textit{Relative Error} metric defined as:}
\begin{equation}
	\eta\triangleq\frac{\text{BLER}_{bound} - \text{BLER}_{sim}}{\text{BLER}_{sim}},
\end{equation}
where $\text{BLER}_{bound}$ represents the BLER upper bound and $\text{BLER}_{sim}$ represents the aproximated BLER obtained through Monte Carlo Simulations. Table \ref{tab:tightness} presents the \textit{relative error} values $\eta$ at selected SNR points for different fading channels. The results reveal that {Theorem \ref{coretheorem}} achieves approximately 90\% reduction in relative error compared to {Theorem \ref{coretheorem1}} across all channel models.

\subsection{Upper Bounds Under Different Parameters Setup}
\begin{figure}[t]
	\centering
	\includegraphics[width=	1\linewidth]{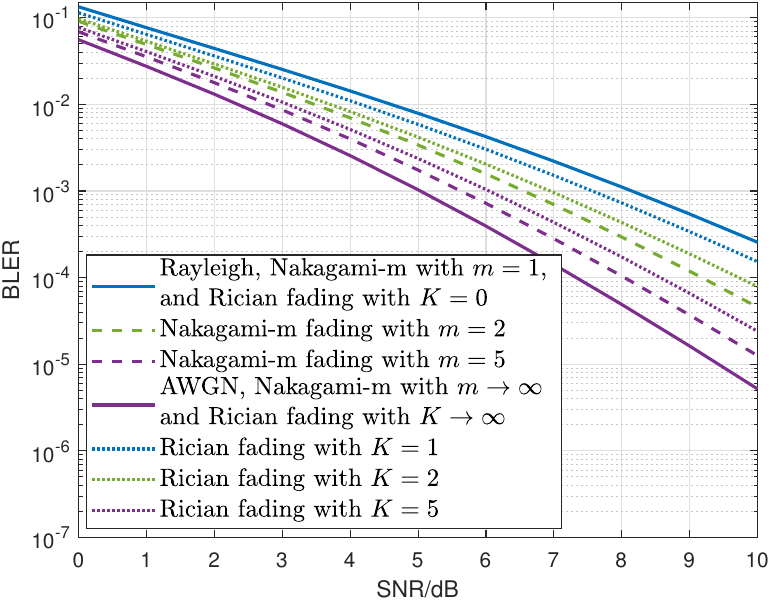}
	\caption{{Upper bounds under different parameters setup. Here $n = 8, v = 32, pass = 6, c = 8, k = 2$ and $\Omega = 1$. The upper bounds are obtained by {Theorem \ref{coretheorem}}.}}
	\label{figure4}
\end{figure}
Fig. \ref{figure4} illustrates various BLER upper bounds under different fading parameter settings, which verifies some important insights:
\subsubsection{\textbf{Nagagami-m Fading Channel}}

{\noindent $\bigcdot$ Fig. \ref{figure4} demonstrates a trend of a decreasing BLER upper bound with increasing $m$ values in the Nakagami-m fading channel, which aligns with the nature of the Nakagami-m model. This trend, which has been revealed in Corollary \ref{coro2}-($i$), suggests that higher $m$ values (indicating less severe fading) result in better channel conditions, leading to lower BLER.}

   {\noindent $\bigcdot$ When $m=1$, Fig. \ref{figure4} demonstrates that the BLER upper bound for the Nakagami-m fading channel overlaps with that of the Rayleigh fading channel. This verifies Corollary \ref{coro2}-($ii$) in this paper.}

    {\noindent $\bigcdot$ When $m\to\infty$, Fig. \ref{figure4} reveals that the Nakagami-m fading channel becomes an AWGN channel. This is consistent with that as $m$ increases, the channel experiences less fading and approaches to an AWGN channel. This verifies Corollary \ref{coro2}-($iii$) in this paper.}
   
   \subsubsection{\textbf{Rician Fading Channel:}}
   {\noindent $\bigcdot$ Fig. \ref{figure4} illustrates that the BLER upper bound decreases with increasing $K$. This is because larger $K$ represents a stronger LOS signal component, which leads to improved channel conditions. This trend has been theoretically revealed in Corollary \ref{coro3}-($i$).}
   
{\noindent $\bigcdot$ When $K=0$, Fig. \ref{figure4} shows that the BLER upper bound for the Rician fading channel overlaps with that of the Rayleigh fading channel. This observation is consistent with Corollary \ref{coro3}-($ii$).}

{\noindent $\bigcdot$ When K approaches infinity, the LOS signal component dominates the received signal. In this case, Fig. \ref{figure4} shows that the bound converges to its infimum, representing an ideal AWGN channel. This trend supports the insight in Corollary \ref{coro3}-($iii$).}

\section{Conclusion} \label{sectionVI}

\textcolor{black}{This paper has derived two explicit upper bounds on the BLER of Spinal codes over real and complex fading channels. One bound is based on the variant of \textit{Gallager bound}, the other bound customize Spinal codes better and has been proved to be tighter. Leveraging the obtained bound, we have obtained the TTP scheme and have theoretically unveiled the optimality of the TTP for ML-decoded Spinal codes.}

Prospective research avenues could extend to the theoretical BLER analysis of Spinal codes across diverse channel models, decoding algorithms, and scenarios of imperfect channel estimation. Furthermore, with the derived tight explicit bounds, optimizing constellation mapping design will be effective for improving Spinal codes. Additionally, the transmission pattern can be refined in the context of practical decoding algorithms. During validations of the proposed bound, we also noticed an error floor in Spinal codes' upper bounds at high SNR. Therefore, to reveal the reason behind the error floor and explicitly derive it may be an interesting work.
\begin{appendices}
\section{The Properties of Hash and Related Inferences}\label{property}
The hash function is expressed as $\mathcal{H}: \{0,1\}^v \times \{0,1\}^k \rightarrow \{0,1\}^v$. It introduces two properties as follows.
\begin{property} \label{property1}
 {The hash function employed by Spinal codes must satisfy the pairwise independence property, as established in \cite{2012Spinal}:}
\begin{equation}
	\begin{split}
		& \mathrm{Pr} \{ \mathcal{H}(\mathbf{s},\mathbf{m}) = \mathbf{x},\mathcal{H}(\mathbf{s}',\mathbf{m}') = \mathbf{x}' \} \\
		&= \mathrm{Pr} \{ \mathcal{H}(\mathbf{s},\mathbf{m}) = \mathbf{x}\} \cdot \mathrm{Pr}\{ \mathcal{H}(\mathbf{s}',\mathbf{m}') = \mathbf{x}' \} \\
		&= 2^{-2v} ,
	\end{split}
\end{equation}
where $(\mathbf{s},\mathbf{m}) \neq (\mathbf{s}',\mathbf{m}')$.
\end{property}
	
\begin{property} \label{property2}
{A sufficiently large $v$ ensures negligible hash collision probability for distinct inputs:} 
	\begin{equation}\label{72}
		\begin{split}
			& \mathrm{Pr} \{ \mathcal{H}(\mathbf{s},\mathbf{m}) = \mathcal{H}(\mathbf{s}',\mathbf{m}') \} \\
			&= \sum_{\mathbf{x} \in {\{ 0,1 \}}^v} \underbrace{\mathrm{Pr} \{ \mathcal{H}(\mathbf{s},\mathbf{m}) = \mathbf{x},\mathcal{H}(\mathbf{s}',\mathbf{m}') = \mathbf{x} \}}_{\text{{Property} }  \ref{property1}} \\
			&= 2^v \cdot 2^{-2v} = 2^{-v} \approx 0,\ \text{iff}\ v \gg 0 ,
		\end{split}
	\end{equation}
where $(\mathbf{s},\mathbf{m})\ne(\mathbf{s}',\mathbf{m}')$ are arbitrary hash inputs.
\end{property}
 {{Property \ref{property2}} is fundamental to the reliability of Spinal codes, as it ensures that distinct inputs generate distinct spine values with high probability when v is sufficiently large, thus enabling effective error correction capbilities.}
\begin{Lemma}
	{For identical inputs and spine values at position $i-1$, the subsequent spine value at position $i$ must also be identical, formally:}
    If $\mathbf{m}_i^* = \mathbf{m}'_i$ and $ \mathbf{s}^*_{i-1} = \mathbf{s}'_{i-1}$, then $\mathbf{s}^*_i = \mathbf{s}'_i, \forall i \in [n/k]$.
\end{Lemma}
\noindent \textit{Proof.}
	From (\ref{eqhash}), we could iteratively prove this lemma. 
\finished
\section{Proof of Theorem \ref{gallagerbound1} }\label{bbb}
{We will establish the bound by examining two distinct cases:
\begin{itemize}
	\item Case I: Complex fading channels where $\gamma=1$.
	\item Case II: Real fading channels where $\gamma=2$.
\end{itemize}}
\subsection{Case I: Complex Fading Channel}
{The standard \textit{Gallager bound} for a specific fading coefficient $\mathbf{H}=(h_1,\cdots,h_L)\in\mathbb{C}^L$ is given as \cite[\textit{Eq. (10)}]{karadimitrakis2017gallager}:}
\begin{equation}\label{gallager79}
	\begin{aligned}
		&\forall \rho\in[0,1],~\Pr\{\mathcal{E}|\mathbf{H}\}\le\\
		&2^{LR}\cdot\int_{\mathbb{C}^L}\left[\sum_{\boldsymbol{\beta} \in \Psi^L} \mathcal{Q}(\boldsymbol{\beta}) f_{\mathbf{Y}}(\mathbf{y} \mid \boldsymbol{\beta},\mathbf{H})^{1 /(1+\rho)}\right]^{1+\rho} \mathrm{d} \mathbf{y}, 
	\end{aligned}
\end{equation}
{where $\mathcal{E}$ is the decoding error, $R$ is the coding rate, $L$ is the codelength,  $\mathcal{Q}(\boldsymbol{\beta})$ is the distribution of the $L$-length channel input $\boldsymbol{\beta}=(\beta_1,\cdots,\beta_L)\in\Psi^L$, $\beta_i\in\Psi$, and $f_{\mathbf{Y}}(\mathbf{y} \mid \boldsymbol{\beta},\mathbf{H})$ is the distribution of the received symbol $\mathbf{y}\in\mathbb{C}^L$ given the $L$-length channel input $\boldsymbol{\beta}$ and the fading coefficient $\mathbf{H}$. To obtain the overall BLER $\Pr\{\mathcal{E}\}$, one can average over all possible fading coefficients to obtain:}
\begin{equation}\label{eq7v2}
	\Pr\{\mathcal{E}\} = \mathbb{E}_{\mathbf{H}}[\Pr\{\mathcal{E}|\mathbf{H}\}].
\end{equation}
{Substituting (\ref{gallager79}) into \eqref{eq7v2} yields the upper bound:}
\begin{equation}\label{6}
	\resizebox{1\hsize}{!}{$
		\begin{aligned}
			&\forall\rho\in[0,1],~\Pr\{\mathcal{E}\}\le\\&
			2^{LR}\cdot\mathbb{E}_{\mathbf{H}}\left[\int_{\mathbb{C}^L}\left[\sum_{\boldsymbol{\beta} \in \Psi^L} \mathcal{Q}(\boldsymbol{\beta}) f_{\mathbf{Y}}(\mathbf{y} \mid \boldsymbol{\beta},\mathbf{H})^{1 /(1+\rho)}\right]^{1+\rho} \mathrm{d} \mathbf{y}\right].
		\end{aligned}$}
\end{equation}

With \eqref{6} in hand, one can minimize over the parameter $\rho\in[0,1]$ on the right-hand side of \eqref{6} to achieve a tight upper bound. However, 
the optimization over $\rho$ introduces substantial complexity. To simplify the analysis, we follow the approach from \cite[\textit{Theorem 2}]{UEPspinal} to {apply a relaxed bound by setting $\rho=1$ rather than finding the optimal value. While this substitution sacrifices some tightness in the bound, it yields a more tractable expression that remains analytically useful while significantly reducing computational complexity. By setting $\rho=1$, we establish the bound given in \eqref{eq82}, where \eqref{eq82}-$(a)$ establishes since the channel is fast fading flat channel with $y_i=h_i\beta_i+n_i,\forall i\in[L]$, which yields $f_{\mathbf{Y}}(\mathbf{y} \mid {\boldsymbol{\beta},\mathbf{H})}=\prod_{i=1}^L f_{Y}(y_i\mid,\beta_i,h_i)$; \eqref{eq82}-$(b)$ holds by expanding the squared summation using the identity $(\sum_{i}a_i)^2=\sum_{i,j}a_ia_j$ and applying the uniform distribution assumption where $\mathcal{Q}(\beta_i)=2^{-c}, \forall i\in[L]$; \eqref{eq82}-$(c)$ is derived by explicitly computing the expectation with respect to the channel fading coefficient $\mathbf{H}$; \eqref{eq82}-$(d)$ is obtained due to the fact that the channel fading coefficients $h_i$ are \textit{independently and identically distributed} (i.i.d); \eqref{eq82}-$(e)$ is derived from the i.i.d nature of $h_i$ and the characteristics of the fast flat fading channel, where  $y_i=h_i\beta_i+n_i,\forall i\in[L]$. These properties enable the decomposition of the complex joint multidimensional integral into a product of individual single-dimensional integrals; \eqref{eq82}-$(f)$ is obtained from the definition of the expectation and  the expansion of the $L$
-th power of a sum, where: $(\sum_{\alpha\in\Psi} f(\alpha))^L=\sum_{\alpha_1,\cdots,\alpha_L\in\Psi}f(\alpha_1)\cdots f(\alpha_L)$. \eqref{eq82}-$(g)$ is established from the following lemma:} 
\begin{figure*}[t]
	\begin{equation}\label{eq82}
		\begin{aligned}
		&\Pr\{\mathcal{E}\}\overset{\rho=1}{\le}\mathbb{E}_{\mathbf{H}}\left[2^{LR}\cdot\int_{\mathbb{C}}\left[\sum_{\boldsymbol{\beta} \in \Psi^L}\sqrt{f_{\mathbf{Y}}(\mathbf{y} \mid {\boldsymbol{\beta},\mathbf{H})}} \mathcal{Q}(\boldsymbol{\beta})\right]^2 \mathrm{d} \mathbf{y}\right]\\
		&\overset{(a)}{=}\mathbb{E}_{\mathbf{H}}\left[2^{LR}\cdot\int_{\mathbb{C}}\left[\sum_{{\beta_1},\beta_2,\cdots,\beta_L \in \Psi}\prod_{i=1}^{L} \mathcal{Q}({\beta_i})\sqrt{f_Y(y_i \mid {\beta_i},{h}_i)}\right]^2 \mathrm{d} {y}_1\cdots \mathrm{d}y_L\right]\\
		&\overset{(b)}{=}\mathbb{E}_{\mathbf{H}}\left[2^{LR}\cdot
		\underset{\substack{{\beta_1},\beta_2,\cdots,\beta_L \in \Psi\\\alpha_1,\alpha_2,\cdots \alpha_L\in\Psi}}{\sum}\int_{\mathbb{C}^L}\prod_{i}^L \mathcal{Q}(\beta_i)\mathcal{Q}(\alpha_i)
		  \sqrt{f_Y(y_i \mid {\beta_i},{h}_i)f_Y(y_i \mid {\alpha_i},{h}_i)} \mathrm{d} {y}_1\cdots \mathrm{d}y_L\right]\\
		  &\overset{(c)}{=}2^{LR}\cdot
		  \underset{\substack{{\beta_1},\beta_2,\cdots,\beta_L \in \Psi\\\alpha_1,\alpha_2,\cdots \alpha_L\in\Psi}}{\sum}\int_{\mathbb{C}^L}f_{\mathbf{H}}(\mathbf{h})\int_{\mathbb{C}^L}\prod_{i}^L\mathcal{Q}(\beta_i)\mathcal{Q}(\alpha_i)
		  \sqrt{f_Y(y_i \mid {\beta_i},{h}_i)f_Y(y_i \mid {\alpha_i},{h}_i)} \mathrm{d} {y}_1\cdots \mathrm{d}y_L \mathrm{d} \mathbf{h}\\
		  &\overset{(d)}{=}2^{LR}\cdot
		  \underset{\substack{{\beta_1},\beta_2,\cdots,\beta_L \in \Psi\\\alpha_1,\alpha_2,\cdots \alpha_L\in\Psi}}{\sum}\int_{\mathbb{C}^L}\int_{\mathbb{C}^L}\prod_{i}^L\mathcal{Q}(\beta_i)\mathcal{Q}(\alpha_i)
		  f_H(h_i)\sqrt{f_Y(y_i \mid {\beta_i},{h}_i)f_Y(y_i \mid {\alpha_i},{h}_i)} \mathrm{d} {y}_1\cdots \mathrm{d}y_L \mathrm{d} \mathbf{h}\\
		  &\overset{(e)}{=}2^{LR}\cdot
		  \underset{\substack{{\beta_1},\beta_2,\cdots,\beta_L \in \Psi\\\alpha_1,\alpha_2,\cdots \alpha_L\in\Psi}}{\sum}\prod_{i=1}^{L}\left[\mathcal{Q}(\beta_i)\mathcal{Q}(\alpha_i)\int_{\mathbb{C}}\int_{\mathbb{C}}
		  f_H(h_i)\sqrt{f_Y(y_i \mid {\beta_i},{h_i})f_Y(y_i \mid {\alpha_i},{h_i})} \mathrm{d}{y_i} \mathrm{d}{h_i}\right]\\
		  &\overset{(f)}{=}
		  2^{LR}\cdot
		  \left[\sum_{\alpha,\beta\in\Psi}\mathcal{Q}(\beta)\mathcal{Q}(\alpha)\mathbb{E}_H\left[\int_{\mathbb{C}}
		  \sqrt{f_Y(y \mid {\beta},{h})f_Y(y \mid {\alpha},{h})} \mathrm{d}{y} \right]\right]^L\\     &\overset{(g)}{=}2^{LR}\cdot
		  \left[\sum_{\alpha,\beta\in\Psi}\mathcal{Q}(\beta)\mathcal{Q}(\alpha)\mathbb{E}_H\left[\exp\left\{-\frac{|H(\beta-\alpha)|^2}{4\sigma^2}\right\} \right]\right]^L\\
		\end{aligned}
	\end{equation}
\hrulefill
\end{figure*}
{\begin{Lemma}\label{lemma9}
	For complex fading channels $f_Y(y|x,H)$ with AWGN variance $\sigma^2$, the following equality holds:
	\begin{equation}
		\begin{aligned}		\int_{\mathbb{C}}\sqrt{f_Y(y \mid {\beta},{H})f_Y(y \mid {\alpha},{H})} \mathrm{d}{y}=\exp\left\{-\frac{|H(\beta-\alpha)|^2}{4\sigma^2}\right\}.
		\end{aligned}
	\end{equation}
\end{Lemma} 
\noindent\textit{Proof.} See Appendix \ref{appi}.}\finished

{\noindent Thus, the equality \eqref{eq82} accomplished the proof.}\finished

{\subsection{Case II: Real Fading Channel}\label{proofcoro}
Over the real fading channel, we can similarly establish the following lemma:
\begin{Lemma}
	For real fading channels $f_Y(y|x,H)$ with AWGN variance $\sigma^2$, the following equality holds:
	\begin{equation}\label{57}
		\begin{aligned}		\int_{\mathbb{R}}\sqrt{f_Y(y \mid {\beta},{H})f_Y(y \mid {\alpha},{H})} \mathrm{d}{y}=\exp\left\{-\frac{|H(\beta-\alpha)|^2}{8\sigma^2}\right\}.
		\end{aligned}
	\end{equation}
\end{Lemma}
\noindent Substituting \eqref{57} into \eqref{eq4-7}-($g$) accomplishes the proof.}\finished

\section{Proof of (\ref{proof:theobound1})}\label{appendixd} 
First, we establish the following inequality:
	\begin{align}
		&\mathscr{F}(L_a,\sigma,\gamma)=\notag\\ &\quad \sum_{t\in[N]} b_t\left(\sum_{\beta,\alpha \in \Psi} 2^{-2c}\mathbb{E}_{H}\left[\exp\left\{-\frac{|H(\beta-\alpha)|^2}{4 \gamma\sigma^2 \mathrm{sin}^2\theta_t}\right\}\right]\right)^{L_a}\notag \\
		&\stackrel{(a)}{\leq} \sum_{t\in[N]} b_t\left(\sum_{\beta,\alpha \in \Psi} 2^{-2c}\mathbb{E}_{H}\left[\exp\left\{-\frac{|H(\beta-\alpha)|^2}{4\gamma \sigma^2}\right\}\right]\right)^{L_a}\notag \\
		&\stackrel{(b)}{=} \left(\sum_{\beta,\alpha \in \Psi} 2^{-2c}\mathbb{E}_{H}\left[\exp\left\{-\frac{|H(\beta-\alpha)|^2}{4\gamma \sigma^2}\right\}\right]\right)^{L_a}\sum_{t\in[N]} b_t \notag\\
		&\stackrel{(c)}{=} \frac{1}{2} \left(\sum_{\beta,\alpha \in \Psi} 2^{-2c}\mathbb{E}_{H}\left[\exp\left\{-\frac{|H(\beta-\alpha)|^2}{4\gamma \sigma^2}\right\}\right]\right)^{L_a},\label{app:ieq1}
	\end{align}
where the inequality \eqref{app:ieq1}-$(a)$  holds because $\mathrm{sin}^2\theta_t \leq 1$ for all $t \in [N]$, which results in a smaller denominator in the exponential function, thereby increasing the overall expression; the equality \eqref{app:ieq1}-$(b)$ is obtained by factored out the constant term from the summation; \eqref{app:ieq1}-$(c)$ is due to the fact that
\begin{equation}
	\sum_{t\in[N]} b_t=\sum_{t\in[N]} \frac{\theta_t-\theta_{t-1}}{\pi}=\frac{\theta_N-\theta_0}{\pi}=\frac{1}{2}.
\end{equation}
Substituting \eqref{app:ieq1} into \eqref{epsilona} yields the following inequality:
\begin{align}
	\epsilon_a^{\text{Spinal}} &= \mathrm{min} \left\{ 1,\left(2^k-1\right)2^{n-ak} \cdot \mathscr{F} \left(L_a , \sigma,\gamma \right) \right\}\notag\\
	&\stackrel{(a)}{\leq} \left(2^k-1\right)2^{n-ak} \cdot \mathscr{F} \left(L_a , \sigma,\gamma \right)\notag\\
	&\stackrel{(b)}{\leq} \left(2^k-1\right)2^{n-ak} \cdot \frac{1}{2} \biggl(\sum_{\beta,\alpha \in \Psi} 2^{-2c} \notag\\
	&\quad\quad\quad \times \mathbb{E}_{H}\Bigl[\exp\Bigl\{-\frac{|H(\beta-\alpha)|^2}{4\gamma \sigma^2}\Bigr\}\Bigr]\biggr)^{L_a}\stackrel{(c)}{=} \epsilon_a^{\text{Mid}}\label{app:final},
\end{align}
where the inequality \eqref{app:final}-$(a)$ establishes because $\min\{x,y\}\le y$; the inequality is derived from \eqref{app:ieq1}; and the equality \eqref{app:final}-$(c)$ is from the definition given in \eqref{eq:tighter}. \finished

\section{Proof of Lemma \ref{Lemma3} }\label{app:prooflemma1}
	Solving $\mathrm{Pr} \left( \mathfrak{R}\left[\mathbf{v}^{L_a} {\left(\mathbf{v}^{L_a} + 2{\mathbf{N}}^{L_a}\right)}^{\mathrm{H}}\right] \leq 0 \right)$ is challenging due to the high dimensionality of $\mathbf{v}^{L_a}$ and $\mathbf{N}^{L_a}$. To simplify, we introduce an $L_a\times L_a$ \textit{unitary matrix} $\mathbf{A}$ to rotate these vectors into a lower-dimensional space. $\mathbf{A}$, defined in $\mathbb{C}^{L_a \times L_a}$, satisfies the unitary condition $\mathbf{A}^{\mathrm{H}} \mathbf{A} = \mathbf{I}_{L_a}$.
	Without loss of generality, we assume $\mathbf{A}$ satisfies that
	\begin{equation} \label{eq4-13}
		\mathbf{A}{\left[\mathbf{v}^{L_a}\right]}^{\mathrm{H}} = {\bigg[ \left\|\mathbf{v}^{L_a}\right\|_2,\underbrace{0,\cdots,0}_{L_a-1} \bigg]}^{\mathrm{T}},
	\end{equation}
	which indicates that the \emph{unitary matrix} $\mathbf{A}$ rotates the vector $\mathbf{v}^{L_a}$ to the direction of a \textit{standard basis}.
	Leveraging $\mathbf{A}^{\mathrm{H}} \mathbf{A} = \mathbf{I}_{L_a}$, the probability of interest $\mathrm{Pr} \left( \mathfrak{R}\left[\mathbf{v}^{L_a} {\left(\mathbf{v}^{L_a} + 2{\mathbf{N}}^{L_a}\right)}^{\mathrm{H}}\right] \leq 0  \right)$ can be transformed as:
	\begin{align}\label{eq4-14}
		&\mathrm{Pr} \left( \mathfrak{R}\left[ \mathbf{v}^{L_a} \mathbf{I}_{L_a} {\left(\mathbf{v}^{L_a} + 2{\mathbf{N}}^{L_a}\right)}^{\mathrm{H}}\right] \leq 0 \right) \notag\\
		= &\mathrm{Pr} \left( \mathfrak{R}\left[ \mathbf{v}^{L_a} {\mathbf{A}}^{\mathrm{H}} \mathbf{A} {\left(\mathbf{v}^{L_a} + 2{\mathbf{N}}^{L_a}\right)}^{\mathrm{H}} \right] \leq 0 \right) \notag \\
		= &\mathrm{Pr} \left(\mathfrak{R}\left\{ {\left[ \mathbf{A}{\left[\mathbf{v}^{L_a}\right]}^{\mathrm{H}} \right]}^{\mathrm{H}} \left( \mathbf{A}{\left[\mathbf{v}^{L_a}\right]}^{\mathrm{H}} + 2\mathbf{A}{\left[\mathbf{N}^{L_a}\right]}^{\mathrm{H}} \right)\right\} \leq 0 \right).
	\end{align}
	Substitute (\ref{eq4-13}) into the RHS of (\ref{eq4-14}), we have
	\begin{equation} \label{eq4-15}
		\begin{split}
			&\mathrm{Pr} \left( \mathfrak{R}\left\{ {\left\| \mathbf{v}^{L_a} \right\|}^2 + 2\left\| \mathbf{v}^{L_a} \right\|_2 \cdot \mathbf{A}_1 {\left[\mathbf{N}^{L_a}\right]}^{\mathrm{H}}\right\} \leq 0 \right) \\&= \mathrm{Pr} \left( \mathfrak{R}\left\{ \mathbf{A}_1 {\left[\mathbf{N}^{L_a}\right]}^{\mathrm{H}} \right\} \leq -\frac{\left\| \mathbf{v}^{L_a} \right\|_2}{2} \right) .
		\end{split}
	\end{equation}
	
	To simplify the RHS of (\ref{eq4-15}), we next introduce another tool, named the \textit{isotropic} properties of random vectors. 
	\begin{Lemma}\label{lemma7}
		\cite[\textit{A. 26, Page 502}]{fundamantalWC}. If $\mathbf{w}\sim\mathcal{CN}(0,\sigma^2\mathbf{I}_{L_a})$, then $\mathbf{w}$ is \textit{isotropic}, \emph{i.e.}, $\mathbf{U} {\left[\mathbf{N}^{L_a}\right]}^{\mathrm{H}} \sim {\left[\mathbf{N}^{L_a}\right]}^{\mathrm{H}}$ for any \textit{unitary matrix} $\mathbf{U}\in\mathbb{C}^{L_a\times L_a}$.
	\end{Lemma}
	\begin{corollary}
		$\mathbf{A}_1 {\left[\mathbf{N}^{L_a}\right]}^{\mathrm{H}}\sim\mathcal{C}\mathcal{N}(0,\sigma^2)$.
	\end{corollary}
	\noindent\textit{{Proof.}}
		As ${\left[\mathbf{N}^{L_a}\right]}^{\mathrm{H}} \sim \mathcal{CN}(0,\sigma^2\mathbf{I}_{L_a})$ and $\mathbf{A}$ is the \textit{unitary matrix}, it holds from {Lemma \ref{lemma7}} that $\mathbf{A} {\left[\mathbf{N}^{L_a}\right]}^{\mathrm{H}} \sim {\left[\mathbf{N}^{L_a}\right]}^{\mathrm{H}}\sim \mathcal{CN}(0,\sigma^2\mathbf{I}_{L_a})$. Since $\mathbf{A_1}$ is the first row of the unitary matrix $\mathbf{A}$, the product $\mathbf{A}_1 {\left[\mathbf{N}^{L_a}\right]}^{\mathrm{H}}$ is also the first row of $\mathbf{A} {\left[\mathbf{N}^{L_a}\right]}^{\mathrm{H}}$, following the distribution $\mathcal{C}\mathcal{N}(0,\sigma^2)$.
	\finished
	
	Since $\mathbf{A}_1 {\left[\mathbf{N}^{L_a}\right]}^{\mathrm{H}}\sim\mathcal{C}\mathcal{N}(0,\sigma^2)$, we could rewrite $\mathbf{A}_1 {\left[\mathbf{N}^{L_a}\right]}^{\mathrm{H}}$ as $\mathbf{A}_1 {\left[\mathbf{N}^{L_a}\right]}^{\mathrm{H}}=W_R+jW_I$, with $W_R, W_I\sim\mathcal{N}(0,\sigma^2/2)$. Therefore, the RHS of (\ref{eq4-15}) can be further simplified by:
	\begin{equation} \label{eq4-21}
		\begin{split}
			\mathrm{Pr} \left(  W_R  \leq -\frac{\left\| \mathbf{v}^{L_a} \right\|_2}{2} \right) = Q\left(\frac{ \left\| \mathbf{v}^{L_a} \right\|_2 }{\sqrt{2}\sigma}\right).
		\end{split}
	\end{equation}
	We thus accomplish the proof of {Lemma \ref{Lemma3}}.\finished
\section{}\label{iidv}
\subsection{Independence and Identically Distributed (i.i.d) $V_{i,j}$}
To establish the i.i.d property of $V_{i,j} = h_{i,j}(\operatorname{f}(\mathbf{x}_{i,j}(\mathbf{M}^*)) - \operatorname{f}(\mathbf{x}_{i,j}(\mathbf{M'})))$, we consider each component separately. The i.i.d nature of $h_{i,j}$ is inherent in the flat fast fading scenario. For $\operatorname{f}(\mathbf{x}_{i,j}(\mathbf{M}^*))$ and $\operatorname{f}(\mathbf{x}_{i,j}(\mathbf{M'}))$, their i.i.d characteristics arise from the RNGs and hash functions used. The RNG, with a consistent seed $\mathbf{s}_i$, ensures the independence of generated symbols, leading to the following lemma.
\begin{Lemma}\label{lemma11-1}
	For $\forall j\ne m$, and $\mathbf{M}$, the encoded symbol $\operatorname{f}(\mathbf{x}_{i,j}(\mathbf{M}))$ is independent of $\operatorname{f}(\mathbf{x}_{i,m}(\mathbf{M}))$.
\end{Lemma}
  The hash function $\mathcal{H}$ ensures that the input and output of the function are independent with each other, \emph{i.e.}, $\forall 1\le i\le n/k-1$, $\mathbf{s}_i$ is independent with $\mathbf{s}_{i+1}$. In this manner, due to the iteration structure with $\mathbf{s}_i = \mathcal{H}(\mathbf{s}_{i-1},\mathbf{m}_i)$, we have
  \begin{Lemma}\label{lemma11}
  	For $\forall i\ne j$, it holds that $\mathbf{s}_i$ is independent of $\mathbf{s}_j$.
  \end{Lemma}
  With {Lemma \ref{lemma11}} in hand, it is then very natural to establish the following corollary:
  \begin{Lemma}\label{lemma13}
  	For $\forall i\ne j$ and $\mathbf{M}$, the symbol $\operatorname{f}(\mathbf{x}_{i,m}(\mathbf{M}))$ is independent of $\operatorname{f}(\mathbf{x}_{j,m}(\mathbf{M}))$
  \end{Lemma}
Combing {Lemma \ref{lemma11-1}} and {Lemma \ref{lemma13}} together leads to the i.i.d characteristics of $\operatorname{f}(\mathbf{x}_{i,j}(\mathbf{M}^*))$ and $\operatorname{f}(\mathbf{x}_{i,j}(\mathbf{M'}))$.
\subsection{Independence Between $h_{i,j}$, $\operatorname{f}(\mathbf{x}_{i,j}(\mathbf{M}^*))$, and $\operatorname{f}(\mathbf{x}_{i,j}(\mathbf{M}'))$}
Since $h_{i,j}$ is the channel coefficient, it is natural to obtain that $h_{i,j}$ is independent with both $\operatorname{f}(\mathbf{x}_{i,j}(\mathbf{M}^*))$, and $\operatorname{f}(\mathbf{x}_{i,j}(\mathbf{M}'))$. The remaining issue is to establish that $\operatorname{f}(\mathbf{x}_{i,j}(\mathbf{M}^*))$ is independent with $\operatorname{f}(\mathbf{x}_{i,j}(\mathbf{M}'))$ for $i\ge a$ and $j\in [\ell_i]$, where $a\triangleq\min\left\{i|\mathbf{s}^*_i=\mathbf{s}'_i\right\}$.
	Before proving this, we first introduce the following lemma:
\begin{Lemma}\label{yn2}
	If $ \mathbf{m}_a^*\ne\mathbf{m}'_{a}$, then $\forall a\le i\le n/k, \Pr\left\{\mathbf{s}^*_{i}=\mathbf{s}'_{i}\right\}\le1-(1-2^{-v})^{i-a+1}$.
\end{Lemma}
\noindent\textit{Proof.}
	Denote $\mathcal{C}_i$ as the event that $\bigcap_{j=a}^{i}\mathbf{s}_{j}^*\ne\mathbf{s}'_{j}$, \emph{i.e.}, $\mathcal{C}_i\triangleq \left\{\bigcap_{j=a}^{i}\mathbf{s}_{j}^*\ne\mathbf{s}'_{j}\right\}$. Then,  by leveraging (\ref{72}), we have
	\begin{equation}
	\begin{aligned}
	\Pr\{\mathcal{C}_a\}&=1-2^{-v},\\
	\Pr\left(\mathcal{C}_{i+1}\left|\mathcal{C}_{i}\right.\right)&=1-2^{-v},\text{for }\forall a\le i\le n/k-1,
	\end{aligned}
	\end{equation}
 Therefore, since $\mathcal{C}_i\subset \mathcal{C}_{i-1}$, we have the recursive relation:
	\begin{equation}\label{6662}
	\begin{aligned}
	\Pr\left(\mathcal{C}_i\right)&=\Pr\left(\mathcal{C}_i\bigcap \mathcal{C}_{i-1}\right)=\Pr\left(\mathcal{C}_i\left|\mathcal{C}_{i-1}\right.\right)\Pr\left(\mathcal{C}_{i-1}\right)\\&=\left(1-2^{-v}\right)\Pr\left(\mathcal{C}_{i-1}\right).
	\end{aligned}
	\end{equation}
	By iteratively leveraging the above recursive relation and $\Pr\{\mathcal{C}_a\}=1-2^{-v}$, we have 
	\begin{equation}\label{75}
	\Pr\left(\mathcal{C}_i\right)=(1-2^{-v})^{i-a+1}.
	\end{equation}
	From (\ref{72}) we know that 
	\begin{equation}\label{76}
	\Pr\left(\mathbf{s}_{i}^*=\mathbf{s}'_{i}\left|\mathcal{C}_{i-1}\right.\right)=\Pr\left(\mathbf{s}_{i}^*=\mathbf{s}'_{i}\left|\mathbf{s}_{i-1}\ne\mathbf{s}'_{i-1}\right.\right)=2^{-v}.
	\end{equation}
	We then establish the recursive inequality relationship between $\Pr\left\{\mathbf{s}^*_{i}=\mathbf{s}'_{i}\right\}$ and $\Pr\left\{\mathbf{s}^*_{i-1}=\mathbf{s}'_{i-1}\right\}$ by leveraging the inequality $\Pr\{A\cap B\}\le\Pr\{A\}$ and applying (\ref{75}) and (\ref{76}):
	\begin{equation}\label{vvv}
	\begin{aligned}
	&\Pr\left\{\mathbf{s}^*_{i}=\mathbf{s}'_{i}\right\}=\Pr\left(\mathcal{C}_{i-1},\mathbf{s}^*_{i}=\mathbf{s}'_{i}\right)+\Pr\left(\overline{\mathcal{C}_{i-1}},\mathbf{s}^*_{i}=\mathbf{s}'_{i}\right)\\
	&\le\Pr\left(\mathcal{C}_{i-1}\right)\Pr\left(\mathbf{s}_{i}^*=\mathbf{s}'_{i}\left|\mathcal{C}_{i-1}\right.\right)+\Pr\left(\mathbf{s}_{i-1}=\mathbf{s}'_{i-1}, \mathbf{s}^*_{i}=\mathbf{s}'_{i} \right)\\
	&\le (1-2^{-v})^{i-a}\cdot 2^{-v}+\Pr\left(\mathbf{s}^*_{i-1}=\mathbf{s}'_{i-1}\right).
	\end{aligned}
	\end{equation}
	Iterating (\ref{vvv}) yields the inequality relationship in {Lemma \ref{yn2}}:
	\begin{equation}\resizebox{1\hsize}{!}{$
	\begin{aligned}
		\Pr\left\{\mathbf{s}_{i}^*=\mathbf{s}'_{i}\right\}\le\sum_{j=a}^{i}(1-2^{-v})^{j-a}\cdot 2^{-v}=1-(1-2^{-v})^{i-a+1}.
	\end{aligned}$}
	\end{equation}
\finished

\noindent With {Lemma \ref{yn2}}, we adopt the sandwich theorem and have
\begin{equation}
0\le\lim\limits_{v\rightarrow\infty}\Pr\left\{\mathbf{s}_{i}^*=\mathbf{s}'_{i}\right\}\le\lim\limits_{v\rightarrow\infty}1-(1-2^{-v})^{i-a+1}=0.
\end{equation}
Consequently, it follows that $\lim\limits_{v\rightarrow\infty}\Pr\left(\mathbf{s}^*_{i}=\mathbf{s}'_{i}\right)=0$. Therefore, for all $a\le i\le n/k$, we can assert that $\mathbf{s}_{i}^*\ne\mathbf{s}'_{i}$ in scenarios with sufficiently large $v$. This leads to that for $\forall a\le i\le n/k, \operatorname{f}\left(\mathbf{x}_{i,j}\left(\mathbf{M}'\right)\right)$ is independent of $\operatorname{f}\left(\mathbf{x}_{i,j}\left(\mathbf{M}^*\right)\right)$.

\section{}
{\subsection{Proof of Lemma \ref{theoremnaka}: Nakagami-m Fading}\label{casestudy1}
Over the Nakagami-m fading channel, the PDF of the modulus of $h_{i,j}$ is $f_{R}(r)=\frac{2 m^m}{\Gamma(m) \Omega^m} \cdot r^{2 m-1} \cdot \exp\{{-m r^2}/{\Omega}\}$. Thus, the expectation $\mathbb{E}_{R}\left[\exp\left\{\frac{-{R^2|\beta-\alpha|}^2}{4\gamma \sigma^2 \mathrm{sin}^2\theta_t}\right\}\right]$ can be expanded as:
\begin{align}
	&\mathbb{E}_{R}\left[\exp\left\{\frac{-{R^2|\beta-\alpha|}^2}{4\gamma \sigma^2 \mathrm{sin}^2\theta_t}\right\}\right]\notag\\
	&=\int_0^{\infty}\exp\left\{\frac{-{r^2|\beta-\alpha|}^2}{4\gamma \sigma^2 \mathrm{sin}^2\theta_t}\right\}\frac{2 m^m r^{2 m-1}}{\Gamma(m) \Omega^m}  \exp\left\{\frac{-m r^2}{\Omega}\right\}dr\notag\\
	&\stackrel{(a)}{=}\int_{0}^{\infty} \exp\left\{-\left({z}+\frac{m}{\Omega}\right){r^2}\right\} \cdot \frac{2m^m}{\Gamma(m)\Omega^m} r^{2m-1}  \mathrm{d}r\notag\\
	&\stackrel{(b)}{=}	\frac{m^m}{\Gamma(m){(z\Omega+m)}^m} \underbrace{{\int_{0}^{\infty}} \mathrm{e}^{-t} t^{m-1} \mathrm{d}t}_{\Gamma(m)}=\frac{m^m}{{(z\Omega+m)}^m},\label{naka-m}
\end{align}
where \eqref{naka-m}-($a$) is established by introducing the variable substitution $z = \frac{|\beta-\alpha|^2}{4\gamma{{\sigma}^{2}}{{\sin}^{2}}\theta_t}$; \eqref{naka-m}-($b$) follows from the the variable substitution $t = \left( z+\frac{m}{\Omega} \right)r^2$. Finally, applying $z = \frac{|\beta-\alpha|^2}{4\gamma{{\sigma}^{2}}{{\sin}^{2}}\theta_t}$ in \eqref{naka-m} and we accomplish the proof.}\finished
{\subsection{Proof of Lemma \ref{theoremricecom}: Rician Fading}\label{casestudy2}
	Over the Rician fading channel, the PDF of the modulus of $h_{i,j}$ is $f_{R}(r)=\frac{2(K+1) r}{\Omega \exp\left\{K+\frac{(K+1) r^2}{\Omega}\right\}} I_0\left(2 \sqrt{\frac{K(K+1)}{\Omega}} r\right)$. Thus, the expectation $\mathbb{E}_{R}\left[\exp\left\{\frac{-{R^2|\beta-\alpha|}^2}{4 \gamma\sigma^2 \mathrm{sin}^2\theta_t}\right\}\right]$ can be calculated by \eqref{rician}, 
	\begin{figure*}
		\begin{equation}\label{rician}
				\begin{aligned}
					\mathbb{E}_{R}\left[\exp\left\{\frac{-{R^2|\beta-\alpha|}^2}{4 \gamma\sigma^2 \mathrm{sin}^2\theta_t}\right\}\right]
					&=\int\limits_0^{\infty}\exp\left\{\frac{-{r^2|\beta-\alpha|}^2}{4 \gamma\sigma^2 \mathrm{sin}^2\theta_t}\right\}\frac{2I_0\left(2 \sqrt{\frac{K(K+1)}{\Omega}} r\right)\cdot(K+1) r}{\Omega \exp\left\{K+\frac{(K+1) r^2}{\Omega}\right\}} dr\\
					&\stackrel{(a)}{=}\frac{2(K+1){\int\limits_{0}^{\infty}} {r\exp\left\{- \left( \frac{K+1}{\Omega}+z \right)r^2\right\} I_0 \left( 2\sqrt{\frac{K(K+1)}{\Omega}}r \right) \mathrm{d}r}}{{\Omega} \mathrm{e}^{K}}\\
					&\stackrel{(b)}{=} \frac{2}{\mathrm{e}^K} \sum_{m=0}^{\infty} \frac{K^m(K+1)^{m+1}}{m ! \Gamma(m+1) \Omega^{m+1}} {\int_0^{\infty} r^{2 m+1} \exp \left\{-\left(\frac{K+1}{\Omega}+z\right) r^2\right\} \mathrm{d} r}\\&\stackrel{(c)}{=} \frac{K+1}{e^K({\Omega}z+K+1)}\cdot \sum_{m=0}^{\infty} \frac{1}{m!}{\left( \frac{K(K+1)}{{\Omega}z+K+1} \right)}^{m}\\&\stackrel{(d)}{=}\frac{K+1}{{\Omega}z+K+1} \cdot \mathrm{exp}\left\{ \frac{-K\Omega z}{{\Omega}z+K+1} \right\}
				\end{aligned}
		\end{equation}
		\hrulefill
	\end{figure*}
	where \eqref{rician}-($a$) follows by the substitution  $z = \frac{|\beta-\alpha|^2}{4\gamma{{\sigma}^{2}}{{\sin}^{2}}\theta_t}$; \eqref{rician}-($b$) is derived by expanding the Bessel function:
	\begin{equation}
		I_0(x) = \sum_{m=0}^{\infty} \frac{(-1)^m}{(m!)^2} \left( \frac{x}{2} \right)^{2m};
	\end{equation}
	\eqref{rician}-($c$) is obtained by solving the following integral
	\begin{equation}\label{eq75}
		\begin{aligned}
			&\int_0^{\infty} r^{2 m+1} \exp \left\{-\left(\frac{K+1}{\Omega}+z\right) r^2\right\} \mathrm{d} r\\&\stackrel{(a)}{=}\frac{m^m}{\Gamma(m){(z\Omega+m)}^m} \underbrace{{\int_{0}^{\infty}} \mathrm{e}^{-t} t^{m-1} \mathrm{d}t}_{\Gamma(m)}\\&=\frac{\Gamma(m+1)}{2{\left( z + \frac{K+1}{\Omega} \right)}^{m+1}},
		\end{aligned}
	\end{equation}
where \eqref{eq75}-($a$) is obtained by performing the variable substitution $t = \left( z+\frac{m}{\Omega} \right)r^2$; \eqref{rician}-($c$) follows by applying the infinite series over the exponential function $\exp\left\{x\right\}=\sum_{m=0}^{\infty}\frac{1}{m!}x^m$. Then, substituting $z = \frac{|u|^2}{4\gamma{{\sigma}^{2}}{{\sin}^{2}}\theta_t}$ back into \eqref{rician} accomplishes the proof.}\finished
{\section{Proof of Corollary \ref{coro2}}\label{proof:coro1}
	\subsection{Proof of ($i$)}
	We apply logarithmic on $G_{\text{Naka}}(m)$ and obtain}
	\begin{equation}
		\ln G_{\text{Naka}}(m)=m\ln{{\left(\frac{4\gamma m{\sigma}^2{\sin}^2{\theta_t}}{{\Omega}|\beta-\alpha|^2+4\gamma m{\sigma}^2{\sin}^2{\theta_t}}\right)}},
	\end{equation}
	{The differential of $\ln G_{\text{Naka}}(m)$ satisfies \eqref{17}, where \eqref{17} holds from the basic inequality $\ln(1-x) + x \le 0, \forall x\ge0$. This implies that $G_{\text{Naka}}(m)$ is monotonically decreasing with respect to $m$.}\finished
	\begin{figure*}
			\begin{align} \label{17}
			\frac{d \ln G_{\text{Naka}}(m)}{d m}&=\frac{{\Omega}|\beta-\alpha|^2}{{\Omega}|\beta-\alpha|^2+4\gamma m{\sigma}^2{\sin}^2{\theta_t}}+\ln\left(\frac{4\gamma m\sigma^2\sin^2\theta_t}{{\Omega}|\beta-\alpha|^2+4\gamma m{\sigma}^2{\sin}^2{\theta_t}}\right)\notag\\
			&{=}\frac{{\Omega}|\beta-\alpha|^2}{{\Omega}|\beta-\alpha|^2+4\gamma m{\sigma}^2{\sin}^2{\theta_t}}+\ln\left(1-\frac{{\Omega}|\beta-\alpha|^2}{{\Omega}|\beta-\alpha|^2+4\gamma m{\sigma}^2{\sin}^2{\theta_t}}\right)\stackrel{(a)}{\le}0.
		\end{align}
		\hrulefill
	\end{figure*}
	{\subsection{Proof of ($iii$)}
	The limit $\lim\limits_{m\rightarrow\infty}G_{\text{Naka}}(m)$ is given as:
	\begin{equation}\label{limit}
		\begin{aligned}
			&\lim\limits_{m\rightarrow\infty}G_{\text{Naka}}(m)=\lim\limits_{m\rightarrow\infty}{\left(\frac{4\gamma m{\sigma}^2{\sin}^2{\theta_t}}{{\Omega}|\beta-\alpha|^2+4\gamma m{\sigma}^2{\sin}^2{\theta_t}}\right)}^m\\&=\lim\limits_{m\rightarrow\infty}{\left(1-\frac{{\Omega}|\beta-\alpha|^2}{{\Omega}|\beta-\alpha|^2+4\gamma m{\sigma}^2{\sin}^2{\theta_t}}\right)}^m\\&\stackrel{(a)}{=}
			\exp\left\{-\frac{\Omega|\beta-\alpha|^2}{4\gamma\sigma^2\sin^2\theta_t}\right\},
		\end{aligned}
	\end{equation} 
	where \eqref{limit}-$(a)$ is established by the fundamental limit $\lim\limits_{x\rightarrow\infty}(1-1/x)^x=\exp\{-1\}$.}\finished
	{\section{Proof of Corollary \ref{coro3}}\label{proof:coro2}
	\subsection{Proof of ($i$)}
	Denote $a=\frac{\Omega|\beta-\alpha|^2}{4\gamma{{\sigma}^{2}}{{\sin}^{2}}\theta_t}$, $G_{\text{Rician}}(K)$ can be rewritten as:
	\begin{equation}\label{eq27}
		G_{\text{Rician}}(K)=\frac{K+1}{a+K+1} \cdot \mathrm{exp}\left\{ \frac{-Ka}{a+K+1}\right\}.
	\end{equation}
	The differential of $G_{\text{Rician}}(K)$ satisfies
	\begin{equation}\label{Kdecrease}
			\begin{aligned}
				&\frac{d G_{\text{Rician}}(K)}{d K}=\\&\frac{-[a^2+(K+1)(a^2+a+1)]}{(a+K+1)^3}\cdot \mathrm{exp}\left\{ \frac{-K\Omega z}{{\Omega}z+K+1} \right\}\stackrel{(a)}{<}0,
			\end{aligned}
	\end{equation}
	where \eqref{eq27}-($a$) holds because $a\ge0$ and $K\ge0$. Thus, $G_{\text{Rician}}(K)$ is monotonically decreasing with $K$.\finished
	\subsection{Proof of ($iii$)}
	Denote the first product term in $G_{\text{Rician}}(K)$ as $F_1(K)$ and the second product term as $F_2(K)$, we have
	\begin{equation}\label{26}
		\begin{aligned}			\lim\limits_{K\rightarrow\infty}G_{\text{Rician}}(K)&=\lim\limits_{K\rightarrow\infty}F_1(K)\cdot\lim\limits_{K\rightarrow\infty}F_1(K)\\&\stackrel{(a)}{=}\exp\left\{-\frac{\Omega|\beta-\alpha|^2}{4\gamma\sigma^2\sin^2\theta_t}\right\},
		\end{aligned}
	\end{equation}
	where \eqref{26}-($a$) holds because
	\begin{equation}\label{27}
		\begin{aligned}
			\lim\limits_{K\rightarrow\infty}F_1(K)&=1,
			\lim\limits_{K\rightarrow\infty}F_2(K)&=\exp\left\{-\frac{\Omega|\beta-\alpha|^2}{4\gamma\sigma^2\sin^2\theta_t}\right\}.
		\end{aligned}
	\end{equation}} \finished
	\section{}\label{appi}	
	\begin{figure*}
		\begin{equation}\label{prooflemma9}
			\begin{aligned}		&\int_{\mathbb{C}}\sqrt{f_Y(y \mid {\beta},{H})f_Y(y \mid {\alpha},{H})} \mathrm{d}{y}\\&\stackrel{(a)}{=}\iint_{\mathbb{R}}\sqrt{f_{\mathfrak{R}[Y]}(\mathfrak{R}[y]\mid \beta,H)\cdot f_{\mathfrak{I}[Y]}(\mathfrak{I}[y]\mid \beta,H)}\times\sqrt{f_{\mathfrak{R}[Y]}(\mathfrak{R}[y]\mid \alpha,H)\cdot f_{\mathfrak{I}[Y]}(\mathfrak{I}[y]\mid \alpha,H)} ~\mathrm{d}\mathfrak{R}[y] \mathrm{d}\mathfrak{I}[y]\\
				&\stackrel{(b)}{=}\frac{1}{\pi\sigma^2}\int_{\mathbb{R}}\exp\left\{ -\frac{(\mu-\mathfrak{R}[H\beta])^2+(\mu-\mathfrak{R}[H\alpha])^2}{2\sigma^2}\right\}\mathrm{d}\mu\times\int_{\mathbb{R}}\exp\left\{ -\frac{(\nu-\mathfrak{I}[H\beta])^2+(\nu-\mathfrak{I}[H\alpha])^2}{2\sigma^2}\right\}\mathrm{d}\nu\\
				&\stackrel{(c)}{=}\frac{\exp\left\{\frac{-(\mathfrak{R}[H(\beta-\alpha)])^2+(\mathfrak{I}[H(\beta-\alpha)])^2}{4\sigma^2}\right\}}{\pi\sigma^2}\times\int_{\mathbb{R}}\exp\left\{-\frac{\left(\mu-\frac{\mathfrak{R}[H(\alpha+\beta)]}{2}\right)^2}{\sigma^2}\right\}  d\mu\times\int_{\mathbb{R}}\exp\left\{\frac{-\left(\nu-\frac{\mathfrak{R}[H(\alpha+\beta)]}{2}\right)^2}{\sigma^2}\right\}  d\nu
				\\&\stackrel{(d)}{=}\exp\left\{\frac{-(\mathfrak{R}[H(\beta-\alpha)])^2+(\mathfrak{I}[H(\beta-\alpha)])^2}{4\sigma^2}\right\}=\exp\left\{-\frac{|H(\beta-\alpha)|^2}{4\sigma^2}\right\}
			\end{aligned}
		\end{equation}
		\hrulefill
	\end{figure*}
	{The proof is accomplished by \eqref{prooflemma9}, where \eqref{prooflemma9}-($a$) is established by leveraging the independence of $\mathfrak{R}[Y]$ and $\mathfrak{I}[Y]$, which allows us to decompose the joint probability density functions:
	\begin{subequations}
		\begin{align}
			f_Y(y \mid {\beta},{H})=f_{\mathfrak{I}[Y]}(\mathfrak{R}[y]\mid \beta,H)\cdot f_{\mathfrak{I}[Y]}(\mathfrak{I}[y]\mid \beta,H),\\
			f_Y(y \mid {\alpha},{H})=f_{\mathfrak{I}[Y]}(\mathfrak{R}[y]\mid \alpha,H)\cdot f_{\mathfrak{I}[Y]}(\mathfrak{I}[y]\mid \alpha,H);
		\end{align}
	\end{subequations}
	\eqref{prooflemma9}-($b$) follows from the channel model where $\mathfrak{R}[y]=\mathfrak{R}[Hx]+\mathfrak{R}[n]$ and $\mathfrak{I}[y]=\mathfrak{I}[Hx]+\mathfrak{I}[n]$, with $\mathfrak{R}[n],\mathfrak{I}[n]\sim\mathcal{N}(0,\sigma^2/2)$. This produces the conditional distributions:
	 \begin{subequations}\label{9}
	\begin{align}
				f_{\mathfrak{R}[Y]}(\mathfrak{R}[y]\mid x,H)=\frac{1}{\sqrt{\pi\sigma^2}}\exp\left\{-\frac{\left(\mathfrak{R}[y]-\mathfrak{R}[Hx]\right)^2}{\sigma^2}\right\},\\
				f_{\mathfrak{I}[Y]}(\mathfrak{I}[y]\mid x,H)=\frac{1}{\sqrt{\pi\sigma^2}}\exp\left\{-\frac{\left(\mathfrak{I}[y]-\mathfrak{I}[Hx]\right)^2}{\sigma^2}\right\}.
			\end{align}
	\end{subequations}
	We introduce the variables $\mu=\mathfrak{R}[y],\nu=\mathfrak{I}[y]$ for clearer notations;
	\eqref{prooflemma9}-($c$) is established by expanding the quadratic terms in the exponents and factoring out terms that are independent of the integration variables. \eqref{prooflemma9}-($d$) applies the standard Gaussian integral formula
	\begin{align}
		\int_{\mathbb{R}}\exp\left\{-\frac{(x-c)^2}{b}\right\}=\sqrt{\pi b}.
	\end{align}}
\finished
	
\end{appendices}



\bibliographystyle{IEEEtran}
\bibliography{reference}

\IEEEbiographynophoto{Aimin Li}
(Graduate Student Member, IEEE) received his B.S. degree in electronics and information engineering from Harbin Institute of Technology (Shenzhen) in 2020, where
he won the Best Thesis Award. He is currently
pursuing the Ph.D. degree with the Department of Electronic Engineering, Harbin Institute of Technology (Shenzhen). From 2023 to 2024, he has visited the Institute for Infocomm Research (I$^2$R), Agency for Science, Technology, and Research (A*STAR), Singapore. His research interests include aerospace communications, advanced channel coding techniques, information theory, and wireless communications. He has served as a reviewer for IEEE JSAC, IEEE TWC, IEEE TNNLS, IEEE TCCN, IEEE ISIT, etc. He has also served as a session chair for IEEE Information Theory Workshop 2024 and IEEE Globecom 2024.

\IEEEbiographynophoto{Xiaomeng Chen}
(Graduate Student Member, IEEE) received her B.E. degree in electronic and information engineering from the Harbin
Institute of Technology (Shenzhen) in 2023, where
she won the Best Thesis Award. She is currently pursuing her M.S.
degree with the Department of Electronic Engineering, HITSZ. Her research interests include advanced
channel coding techniques, wireless communications, and information theory.

\vspace{-1cm}
\IEEEbiographynophoto{Shaohua Wu} (Member, IEEE) received the Ph.D. degree in communication engineering from the Harbin Institute of Technology (HIT), Harbin, China, in 2009. From 2009 to 2011, he held a Postdoctoral position with the Department of Electronics and Information Engineering, Shenzhen Graduate School, HIT (Shenzhen), Shenzhen, China, where he has been an Associate Professor since 2012. He is also an Associate Professor of Peng Cheng Laboratory, Shenzhen. From 2014 to 2015, he was a Visiting Researcher with BBCR, University of Waterloo, Waterloo, ON, Canada. He holds more than 30 Chinese patents. He has authored or coauthored more than 100 papers in the above-mentioned areas. His current research interests include wireless image/video transmission, space communications, advanced channel coding techniques, and B5G wireless transmission technologies.
\IEEEbiographynophoto{Gary C.F. Lee} (Member, IEEE) received his B.S. degree in Electrical Engineering in Stanford University in 2016, his M.S. degree in Electrical Engineering in Massachusetts Institute of Technology (MIT) in 2019, and his P.hD. degree in Electrical Engineering in Massachusetts Institute of Technology (MIT) in 2023. He is currently a senior scientist in the Institute for Infocomm Research (I$^2$R), Agency for Science, Technology, and Research (A*STAR), Singapore. His current research interest is in wireless communications, machine learning, and signal processing.

\IEEEbiographynophoto{Sumei Sun} (Fellow, IEEE) is a Principal Scientist, Distinguished Institute Fellow, and Acting Executive Director of the Institute for Infocomm Research (I$^2$R), Agency for Science, Technology, and Research (A*STAR), Singapore. She is also holding a joint 
appointment with the Singapore Institute of Technology, and an adjunct appointment 
with the National University of Singapore, both as a full professor. Her current 
research interests are in next-generation wireless communications, cognitive 
communications and networks, industrial internet of things, communications-computing-control integrative design, joint radar-communication systems, and signal 
intelligence. She is Editor-in-Chief of the IEEE Open Journal Vehicular Technology, and 
Steering Committee Chair of the IEEE Transactions on Machine Learning in 
Communications and Networking. She is also Member-at-Large of 
the IEEE Communications Society, a member of the IEEE Vehicular Technology Society Board of Governors (2022-2024), Fellow of the IEEE and the Academy of Engineering Singapore.
\end{document}